\documentclass[journal,12pt,onecolumn,draftclsnofoot]{IEEEtran}
\usepackage{amsmath,amsfonts}
\usepackage{amssymb}
\usepackage{amsthm}
\usepackage{algorithmic}
\usepackage[ruled,vlined]{algorithm2e}
\usepackage{array}
\usepackage{bbm}
\usepackage{color}
\usepackage[caption=false,font=normalsize,labelfont=sf,textfont=sf]{subfig}
\usepackage{cite}
\usepackage{graphicx}
\usepackage{stfloats}
\usepackage{textcomp}
\usepackage{url}
\usepackage{verbatim}

\begin{document}

\newtheorem{definition}{Definition}
\newtheorem{assumption}{Assumption}
\newtheorem{lemma}{Lemma}
\newtheorem{theorem}{Theorem}
\newtheorem{proposition}{Proposition}
\newtheorem{remark}{Remark}

\title{Gradient Sparsification for Efficient Wireless Federated Learning with Differential Privacy}


\author{Kang~Wei,
        Jun~Li,
        Chuan~Ma,
        Ming~Ding,
        Feng Shu,\\
        Haitao Zhao,
        Wen~Chen,
        and Hongbo Zhu}

\markboth{Submitted to SCIS}%
{Shell \MakeLowercase{\textit{et al.}}: A Sample Article Using IEEEtran.cls for IEEE Journals}
\maketitle
\begin{abstract}
Federated learning (FL) enables distributed clients to collaboratively train a machine learning model without sharing raw data with each other. However, it suffers from the leakage of private information from uploading models. In addition, as the model size grows, the training latency increases due to the limited transmission bandwidth and model performance degradation while using differential privacy (DP) protection. In this paper, we propose a gradient sparsification empowered FL framework with DP over wireless channels, in order to improve training efficiency without sacrificing convergence performance. Specifically, we first design a random sparsification algorithm to retain a fraction of the gradient elements in each client's local model, thereby mitigating the performance degradation induced by DP and reducing the number of transmission parameters over wireless channels. Then, we analyze the convergence bound of the proposed algorithm, by modeling a non-convex FL problem. Next, we formulate a time-sequential stochastic optimization problem for minimizing the developed convergence bound, under the constraints of transmit power, the average transmitting delay, as well as the client's DP requirement. Utilizing the Lyapunov drift-plus-penalty framework, we develop an analytical solution to the optimization problem. Extensive experiments have been implemented on three real-life datasets to demonstrate the effectiveness of our proposed algorithm. We show that our proposed algorithms can fully exploit the interworking between communication and computation to outperform the baselines, i.e., random scheduling, round robin and delay-minimization algorithms.
\end{abstract}


\section{Introduction}
\IEEEPARstart{F}{ederated} learning (FL) has emerged as a new distributed learning paradigm that enables multiple clients to collaboratively train a shared model without sharing their local data~\cite{Nguyen2021Enabling,Kairouz2019Advances}.
However, FL faces several critical challenges, that is, the storage, computational, and communication capabilities equipped at each client may differ due to variability in the computation frequency, memory, limited bandwidth, and power constraints. In light of this, the well-known FedAvg algorithm with local stochastic gradient descent (SGD) and partial participation of clients is widely adopted to reduce the training latency and communication overhead~\cite{Ma2019FL}.
Furthermore, several improved FL algorithms over wireless channels have been proposed to lower the inter-client variance caused by data heterogeneity and device heterogeneity~\cite{Xia2020Multi}.

Recently, in order to improve the training efficiency, some works proposed effective scheduling algorithms for wireless and computation resources between clients and the edge server in FL~\cite{Yang2020Scheduling,Wang2019Adaptive,Chen2020A,Chen2020Convergence,yang2021Energy}.
The work in~\cite{Yang2020Scheduling} characterized the performance of FL in wireless networks and developed an analytical model on FL convergence rate to evaluate the effectiveness of three different client scheduling schemes, i.e., random scheduling, round robin, and proportional fair.
A control algorithm was proposed to minimize the loss function of the convergence bound of distributed SGD by formulating FL training over a wireless network as an optimization problem~\cite{Wang2019Adaptive}.
Further, the works in~\cite{Chen2020A} and~\cite{Chen2020Convergence} improved the training efficiency based on the convergence bound by optimizing the client selection and power allocation via constructing the connection between the wireless resource allocation and the FL training performance using the convergence bound.
The work in~\cite{yang2021Energy} formulated the FL training and communication problem as an optimization problem to minimize the total energy consumption of the system under a latency constraint.
To address this problem, the work in~\cite{yang2021Energy} provided an iterative algorithm and closed-form solutions at every step for time allocation, bandwidth allocation, power control, computation frequency, and learning accuracy.

Different from effective scheduling algorithms for wireless and computation resources~\cite{Deng2022Blockchain,Deng2022Low}, model/gradient compression techniques such as sparsification~\cite{Decebal2019Scalable} and quantization~\cite{Jiang2018SketchML} can be an alternative method to balance training performance and transmission delay~\cite{Chen2021Communication}.
The work in~\cite{Zheng2021Design} introduced the finite-precision quantization in uplink and downlink communications, and provided new convergence analysis of the well-known federated average (FedAvg) in the non-independent and identically distributed (non-IID) data setting and partial clients participation.
Theoretical results revealed that, with a certain quantization, transmitting the weight differential can achieve a faster convergence rate, compared with transmitting the weight.
The work in~\cite{Wang2022Quantized} analyzed the FL convergence in terms of quantization errors and the transmission outage, and then concluded that its performance can be improved if the clients have uniform outage probabilities.
The work in~\cite{Liu2022Joint} involved model pruning for wireless FL and formulated an optimization problem to maximize the convergence rate under the given learning latency constraint.
These works can reduce the transmission overhead efficiently by compressing the uploading models, but not consider the possible privacy leakage caused by exchanged learning models.
Thus, it is desirable to utilize some privacy computing techniques to protect the shared learning models
with the rigorous privacy guarantee~\cite{Wang2019Beyond}, such as differential privacy (DP)~\cite{Dwork2014The}.

Differentially private FL provides a rigorous privacy guarantee for clients' train data, in which there exists an inherent tradeoff between model performance and data privacy~\cite{Wei2020Fed}.
This is due to the implementation of the DP definition, which requires bounding the influence of each example on the gradient from SGD training or local models, and then perturbing it by random noises with a certain scale~\cite{Wei2021User}.
Especially for larger networks, they suffer from far greater distortions during training compared to their non-private counterparts, which results in significant penalties to utility~\cite{Abadi2016Deep}.
To alleviate such degradation, one approach for the client is to leverage the gradient sparsification technique, which can keep the stochastic gradients stay in a low dimensional subspace during training~\cite{Zhou2021Bypassing}.
However, existing works for DP based FL systems have not jointly considered the advantage of the gradient sparsification from both communication efficiency and DP training aspects, and designed an efficient framework to improve the training performance and communication efficiency~\cite{Yang2021Privacy,Wei2022Low}.

To fill this gap, we propose a novel differentially-private FL scheme in wireless networks, termed differentially private FL with gradient sparsification (DP-SparFL), to provide a low communication overhead while maintaining a high model accuracy under the required privacy guarantee.
The main contributions of this paper can be summarized as follows:
\begin{itemize}
\item[$\bullet$] We propose a gradient sparsification empowered communication-efficient FL system with the DP guarantee, which will randomly reduce the elements in the gradients of each local training.
    We also improve the DP based FL performance using an adaptive gradient clipping technique with configurable gradient sparsification rates.
    Our proposed algorithm can efficiently reduce the detrimental influence of DP on the training and communication overhead in wireless networks.
\item[$\bullet$] To further improve the training efficiency, we analyze the convergence bound in terms of gradient sparsification rates in the non-convex FL setting, which is more general than the convex problems.
    We formulate a novel stochastic optimization problem that minimizes the newly found convergence bound, while satisfying transmit power, average delay and client's DP requirement constraints.
    Using the Lyapunov drift-plus-penalty framework, we provide an analytical and feasible solution to the problem.
\item[$\bullet$] Extensive experimental results, for three classification tasks, including MNIST, Fashion-MNIST and CIFAR-$10$ datasets, have been provided to demonstrate the effectiveness of our proposed algorithm.
    Moreover, we show that the proposed algorithm outperforms the baselines, i.e., random scheduling, round robin and delay-minimization algorithms, with the DP requirement.
\end{itemize}

The rest of this paper is organized as follows.
We introduce background on the FL framework and the concept of DP in~Section~\ref{sec:Preliminaries}.
In Section~\ref{sec:Sys_model}, we illustrate the system model of wireless FL and propose the DP-SparFL algorithm.
Then, we analyze the convergence bound in terms of the gradient sparsification rate, and formulate the joint channel assignment problem as an optimization problem with feasible solutions in Section~\ref{sec:Client_scheduling}.
Experimental results are described in Section~\ref{sec:Exm_res}.
Finally, conclusions are drawn in Section~\ref{sec:Concl}.

\section{Preliminaries}\label{sec:Preliminaries}
In this section, we will present preliminaries and related background knowledge on FL, DP and the gradient sparsification technique.
\subsection{Federated Learning}
We consider a wireless FL system where an access point (AP) with a central server and $N$ available channels is located in the center of a wireless network and $U$ clients are randomly distributed within the coverage of AP.
In order to train a machine learning (ML) model, AP needs to fully utilize all clients' data during $T$ communication rounds.
Let $\mathcal{D}_{i}$ denote the local dataset held by the $i$-th client, where $i \in \mathcal{U}$ and $\mathcal{U}=\{1,2,\ldots,U\}$ is the index set of clients.
The overall training process of such a wireless FL system can be divided into four parts as follows: 
1) AP transmits the global model $\boldsymbol{w}^{t}$ and training information (e.g., channel assignment) to selected clients at the $t$-th communication round; 2) the selected clients perform local training to update the local models based their local datasets; 3) the selected clients upload their local models/ updates to AP; 4) AP aggregates all received local models/updates to generate a new global model.
The aggregation process  at the $t$-th communication round in the central server can be expressed as
\begin{equation}
\Delta\boldsymbol{w}^{t}=\sum\limits_{i\in \mathcal{S}^{t}}p_{i}^{t}\Delta\boldsymbol{w}_{i}^{t},
\end{equation}
where $\mathcal{S}^{t}$ is the set of selected clients ($\mathcal{S}^{t} \subseteq \mathcal{U}$), $p_{i}^{t}=\vert \mathcal{D}_{i}\vert/\sum_{i\in \mathcal{S}^{t}}\vert \mathcal{D}_{i}\vert$, $\vert \mathcal{D}_{i}\vert$ is the size of $\mathcal{D}_{i}$, $\Delta\boldsymbol{w}^{t}$ is the global update at the $t$-th communication round, $\Delta\boldsymbol{w}_{i}^{t}$ is the local update of the $i$-th client, given by $\Delta\boldsymbol{w}_{i}^{t}=\boldsymbol{w}_{i}^{t-1,\tau} - \boldsymbol{w}^{t-1}$,
$\boldsymbol{w}_{i}^{t-1,\tau}$ is the local model after $\tau$ local training epochs and $\boldsymbol{w}^{t-1}$ is the global model at the $(t-1)$-th communication round.
Overall, we can formulate the FL task as
\begin{equation}
\boldsymbol{w}^{\star}=\mathop{\arg\min}_{\boldsymbol{w} }{F(\boldsymbol{w})},
\end{equation}
where $F(\boldsymbol{w})=\sum_{i\in\mathcal{U}}p_{i}F_{i}(\boldsymbol{w})$
represents the global loss function, $F_{i}(\cdot)$ is the local loss function of the $i$-th client and $p_{i} =\vert \mathcal{D}_{i}\vert/\sum_{i\in \mathcal{U}}\vert \mathcal{D}_{i}\vert$.
We can observe that all clients have the same data structure and learn a global ML model collaboratively in the training procedure.
\subsection{Differential Privacy}
DP mechanism has drawn a lot of attention because it can provide a strong criterion for the privacy preservation of distributed learning systems with parameters $\epsilon$ and $\delta$.
Research in differentially private ML models tracks a relaxed variant of DP, known as R{\'{e}}nyi DP (RDP)~\cite{Mironov2019R} that has already been widely adopted, such as Opacus in Facebook and tensorflow privacy in Google~\cite{Yousefpour2021Opacus}.
We first define the neighborhood dataset $\mathcal{D}'$ as adding or removing one record in the dataset $\mathcal{D}$.
Thus, we will adopt the RDP technique for the privacy budget computation, and then define RDP as follows.
\begin{definition}($(\alpha, \epsilon)$-RDP):
Given a real number $\alpha \in (1, +\infty)$ and privacy level (PL) $\epsilon$, a randomized mechanism $\mathcal M$ satisfies $(\alpha, \epsilon)$-RDP for any two adjacent datasets $\mathcal{D}, \mathcal{D}'\in \mathcal{X}$, the R{\'{e}}nyi distance between $\mathcal M(\mathcal{D})$ and $\mathcal M(\mathcal{D}')$ is given by
\begin{equation}
D_{\alpha}[\mathcal M(\mathcal{D})\|\mathcal M(\mathcal{D}')]:= \frac{1}{\alpha-1}\log {\mathbb{E}\left[\left(\frac{\mathcal M(\mathcal{D})}{\mathcal M(\mathcal{D}')}\right)^{\alpha}\right]} \leq \epsilon,
\end{equation}
where the expectation is taken over the output of $\mathcal M(\mathcal{D})$.
\end{definition}
We can note that RDP is a generalization of $(\epsilon, \delta)$-DP that adopts R{\'{e}}nyi divergence as a distance metric between two distributions.
ML models achieve RDP guarantees by two alterations to the training process: clipping the per-sample gradient and adding Gaussian noise to training gradients, as known as DP based SGD (DP-SGD).
In our model, we assume that the server is curious-but-honest and may infer private information of clients by analyzing their local models/updates.
The $i$-th client needs to achieve the $(\epsilon_{i}, \delta)$-DP requirement with a proper Gaussian noise standard deviation (STD) $\widehat{\sigma}$ with DP-SGD training, where $\epsilon_{i}$ is the required PL by the $i$-th client.
\subsection{Gradient Sparsification}
Gradient sparsification~\cite{Stich2018Sparsified} has been widely applied to reduce the model size as well as relieving the high communication burden over wireless channels.
In the gradient sparsification method, for a given gradient sparsification rate (\emph{a
certain percentage of elements that have been retained}) $s$ and a model gradient vector $\boldsymbol{g}$, this method will generate a binary mask $\boldsymbol{m}$ with the same size with $\boldsymbol{g}$.
When using the random sparsifier, the element in $\boldsymbol{m}$ is set to $1$ with the probability $s$, and $0$ with the probability $(1-s)$, respectively.
In each training process, the gradient update rule with gradient sparsification can be expressed as
\begin{equation}
\boldsymbol{g}_{k} =
\begin{cases}
\boldsymbol{g}_{k},& \text{if}\,\,\, \boldsymbol{m}_{k} = 1;\\
 0, & \text{otherwise},
\end{cases}
\end{equation}
where $\boldsymbol{g}_{k}$ is the $k$-th element in the gradient vector $\boldsymbol{g}$, $k\in \{1, \ldots, \vert \boldsymbol{g} \vert \}$.
It can be noted that when the $k$-th element of $\boldsymbol{m}$ is zero, the element in the same position of gradient does not need to update.
Because $\boldsymbol{m}$ is a binary vector, the total number of bits to transmit the sparse local update to the server can be expressed as $32s\vert \boldsymbol{g}\vert + \vert \boldsymbol{g}\vert$, where $\vert \boldsymbol{g}\vert$ is the size of $\boldsymbol{g}$ and a $32$-bit representation is adopted here.
Hence, when $s$ is small, using the sparsification method for each client in FL can efficiently reduce the communication overhead.
\section{System Model}\label{sec:Sys_model}
In this subsection, we first present the system model and problem formulation of the client resource allocation supporting FL over a frequency division wireless environment, as shown in Fig.~\ref{Fig:System_model}.
We assume downlink communication for AP will adopt the broadcast module, i.e., the same channel for all clients.
Due to the power constraint for each client device and limited number of allocated channels to each scheduled device,
we consider that not all clients will upload their updates at the aggregation step.
\begin{figure}[ht]
\centering
\includegraphics[width=3.8in,angle=0]{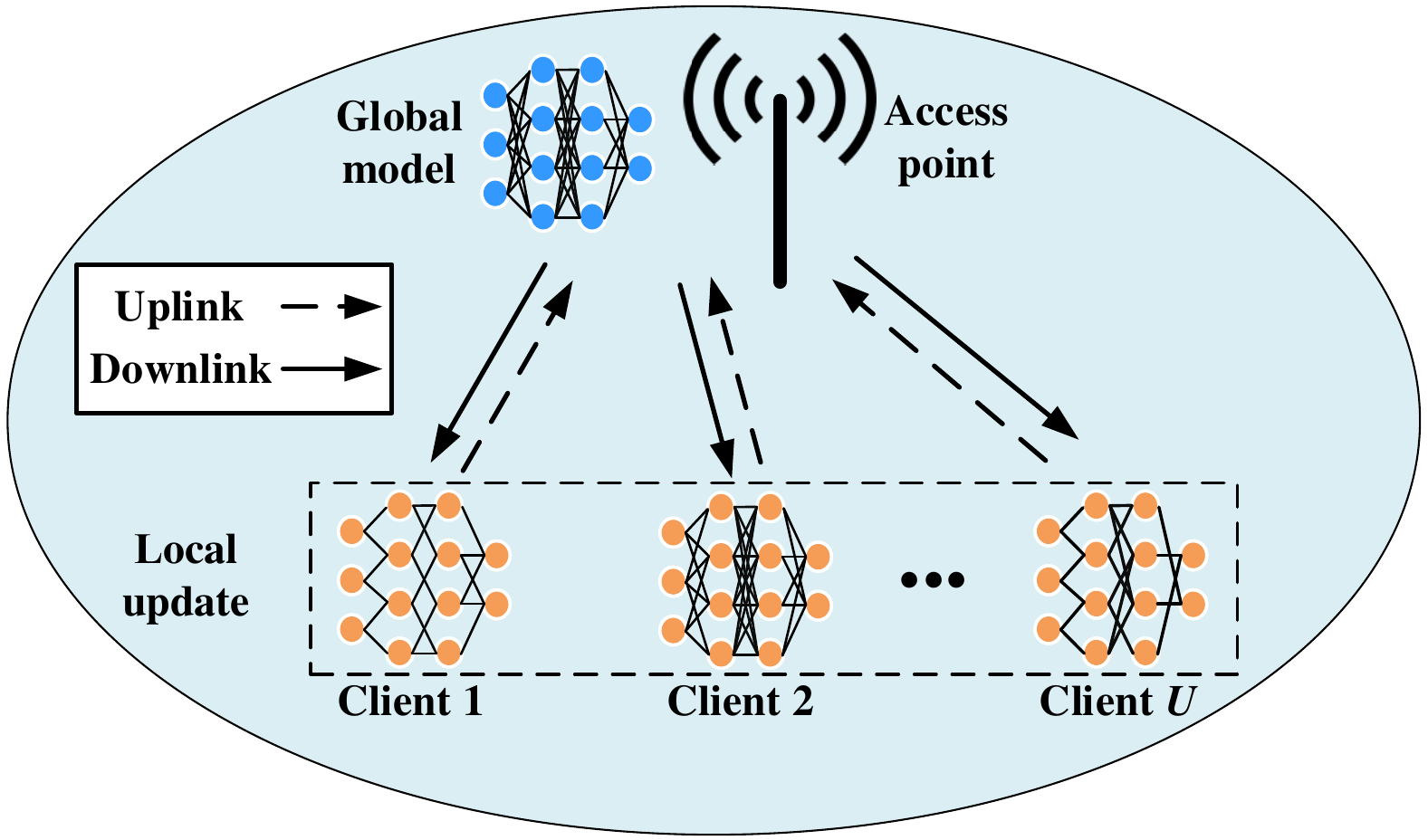}
\caption{The framework of wireless FL consists of multiple clients and an AP with multiple channels, where each client will train a learning model and upload the local update at each communication round.}
\label{Fig:System_model}
\end{figure}
\subsection{Communication Model}
At the beginning of each communication round, the AP broadcasts the global model to selected clients. Hence, the downlink data rate and the transmission latency at the $t$-th communication round can be calculated as
$C_{i}^{t, \text{do}} = B\log_{2}\left(1+\frac{P^{t} h_{i}^{t, \text{do}}}{I^{t, \text{do}}_{i}+\sigma^{2}}\right)$ and $d^{t, \text{do}}_{i}=\frac{Z^{t}}{C_{i}^{t, \text{do}}}$, respectively, 
where $P^{t}$ is the transmit power of the server, $h_{i}^{t, \text{do}}$ is the average channel gain for the downlink channel, $\sigma^{2}$ is the noise power spectral density, $I^{t, \text{do}}_{i}$ is the interference caused by other wireless equipments, $Z^{t}$ is defined as the number of bits that the AP requires to transmit vector $\Delta\boldsymbol{w}^{t}$ over wireless links and $B$ denotes the bandwidth of the uplink channel. The uplink data rate of the $i$-th client transmitting its local model via the $j$-th channel ($j \in \mathcal{N}$ and $ \mathcal{N} = \{1, \ldots, N\}$) to the AP at the $t$-th communication round can be expressed as
$C_{i, j}^{t, \text{up}} = B\log_{2}\Big{(}1+\frac{P_{i}^{t} h_{i, j}^{t}}{I^{t, \text{up}}_{i,j}+\sigma^{2}}\Big{)}$,
where $P_{i}^{t}$ is the transmit power of the $i$-th client, $h_{i, j}^{t}$ is the average channel gain and $I^{t, \text{up}}_{i,j}$ is the co-channel interference caused by the wireless equipments that are located in other service areas.
The transmission delay and energy consumption between the $i$-th client and the $j$-th channel over the uplink channel at the $t$-th communication round can be computed as
$d^{t, \text{up}}_{i, j}=\frac{Z_{i}^{t}}{C_{i, j}^{t, \text{up}}}$ and $E^{t, \text{co}}_{i, j}= P_{i}^{t} d^{t, \text{up}}_{i, j}$, respectively, 
where $Z_{i}^{t}$ is defined as the number of bits to transmit the vector $\Delta\boldsymbol{w}_{i}^{t}$ over wireless links.
\subsection{Computation Model}
In the computation model, each client is equipped with a CPU for the training task, in which the CPU frequency  $f_{i}^{t}$ (in CPU cycle/s) of the $i$-th client is changed at differet communication rounds.
Moreover, the number of CPU cycles performing the forward-backward propagation process for one data of the $i$-th client can be denoted as $\Phi_{i}$.
Due to the fact that CPU operates in the serial mode, the latency of local training can be expressed as
$d^{t, \text{lo}}_{i} = \tau \vert \mathcal D_{i}\vert \Phi_{i}/f_{i}^{t}$,
where $\tau$ is the number of local training epochs.
Then, the CPU energy consumption of the $i$-th client for one local round of computation is expressed as
$E^{t, \text{cp}}_{i} = \chi_{i}\tau \vert \mathcal D_{i}\vert \Phi_{i}\left(f_{i}^{t}\right)^{2}/2$,
where $\chi_{i}$ represents the effective capacitance coefficient of the $i$-th client's computing chipset~\cite{Dinh2021Federated}.
\subsection{Channel Allocation}
Due to the limited channels, we need to allocate the channel resources for clients, i.e., client-channel matching, at each communication round.
We define $\boldsymbol{a}^{t}_{i,j}\in\{0,1\}, i\in \mathcal{U}, j\in \mathcal{N}$ as the matching indictor.
Thus, $\boldsymbol{a}^{t}_{i,j}=1$ represents the $i$-th client will upload its update parameters via the $j$-th channel.
Moreover, the indictor must satisfy the constraint $\sum_{j\in\mathcal{N}}\boldsymbol{a}^{t}_{i,j}\leq 1$ to avoid the channel conflict.
In this way, we can rewrite the aggregation process in terms of channel allocation results as
\begin{equation}\label{equ:model_aggre}
\boldsymbol{w}^{t}=\boldsymbol{w}^{t-1}+\sum\limits_{i\in \mathcal{U}}\sum\limits_{j\in \mathcal{N}}p_{i}^{t}\boldsymbol{a}_{i,j}^{t}\Delta \boldsymbol{w}_{i}^{t}.
\end{equation}
We can observe that $\sum_{j\in\mathcal{N}}\boldsymbol{a}^{t}_{i,j} = 1$ if the $i$-th client has been selected at the $t$-th communication round, and $\sum_{j\in\mathcal{N}}\boldsymbol{a}^{t}_{i,j}=0$ otherwise.
In addition, the set of selected clients at the $t$-th communication round can be represented by $\mathcal{S}^{t}=\{i, \forall i\in\mathcal{U}|\sum_{j\in\mathcal{N}}\boldsymbol{a}^{t}_{i,j} = 1\}$.
\subsection{Differentially Private FL over Wireless Networks}
In Differentially Private FL, we utilize the DP-SGD mechanism, i.e., clipping the gradient with the threshold $C$ and adding Gaussian noise with the variance $\widehat{\sigma}^2 C^2\mathbf{I}$, where $\mathbf{I}$ is an identify matrix, to achieve the DP guarantee.
In addition, we also need to calculate the accumulative privacy budget along with training rounds.
When the accumulative privacy budget $\widehat{\epsilon}_{i}$ of the $i$-th client is up to its preset PL $\epsilon_{i}$ at the next round, this client will be removed from the clients set.
Based on the calculation method in~\cite{Yousefpour2021Opacus}, we can calculate the accumulative privacy budget for the $i$-th client after $\overline{t}$ times of model upload (exposures) as
\begin{equation}\label{eq:accumulative_eps}
\begin{aligned}
\widehat{\epsilon}_{i} = \overline{\epsilon}_{i} + \frac{\log(\frac{1}{\delta})+(\alpha-1)\log(1-\frac{1}{\alpha})-\log(\alpha)}{\alpha-1},
\end{aligned}
\end{equation}
where
$\overline{\epsilon}_{i} = \frac{\overline{t}\tau}{\alpha-1}\ln {\mathbb{E}_{z\sim \mu_{0}(z)}\left[\left(1-q_{i}+q_{i} \mu_{1}(z)/\mu_{0}(z)\right)^{\alpha}\right]}$,
$q_{i} = \vert \boldsymbol{b}\vert/\vert\mathcal{D}_{i}\vert$ is the sample rate in the DP-SGD training, $\vert \boldsymbol{b} \vert$ is the batch size, $\boldsymbol{b}$ is the sample set selected randomly from $\mathcal{D}_{i}$, $\mu_{0}(z)$ and $\mu_{1}(z)$ denote the Gaussian probability density function (PDF) of $\mathcal{N}(0,\widehat{\sigma})$ and the mixture of two Gaussian distributions $q_{i}\mathcal{N}(1,\widehat{\sigma})+(1-q_{i})\mathcal{N}(0,\widehat{\sigma})$, respectively, $\tau$ is the number of local epochs and $\alpha$ is a selectable variable.
Based on the above preparation, the detailed steps of differentially private FL over wireless networks are introduced in \textbf{Algorithm~\ref{alg:DP-FL}}.
\begin{algorithm}[h]
\label{alg:DP-FL}
\caption{Differentially Private FL over Wireless Networks}
\KwIn{Each client owns a local training dataset $\mathcal{D}_i$, an individual privacy parameter\\ $(\epsilon_{i}, \delta)$ and a preset clipping value $C$, the noise STD $\widehat{\sigma}$;}
\KwOut{Trained global model $\boldsymbol{w}^{f}$}
\textbf{Initialization:} Server initializes a global model $\boldsymbol{w}^{0}$, set $\mathcal{\widehat{U}}=\mathcal{U}$; \\
\For{$t = 0, 1, \ldots, T-1$}{
The server selects a client set $\mathcal{S}^{t}$ from $\mathcal{\widehat{U}}$ and completes the resource allocation;\\
The server broadcasts the global model $\boldsymbol{w}^{0}$ to clients in $\mathcal{S}^{t}$; \\
\For{client $i$ in $\mathcal{S}^{t}$ (in parallel)}
{
Update the local model: $\boldsymbol{w}^{t,0}_{i}=\boldsymbol{w}^{t}$; \\
\For{each local epoch $\ell$ from $0$ to $\tau-1$}{
Clip and average gradients of samples in the batch $\boldsymbol{b}$ by $\boldsymbol{g}_{i}^{t, \ell} = \frac{1}{\vert \boldsymbol{b} \vert}\sum_{m \in \boldsymbol{b}} \nabla F(\boldsymbol{w}^{t,\ell}_{i}, \mathcal{D}_{i, m})\cdot\min\{1, \frac{\Vert \nabla F(\boldsymbol{w}^{t,\ell}_{i}, \mathcal{D}_{i, m})\Vert}{C}\}$, where $\mathcal{D}_{i, m}$ is the $m$-th sample in $\mathcal{D}_{i}$ and $ \nabla F(\boldsymbol{w}^{t,\ell}_{i}, \mathcal{D}_{i, m})$ is the gradient for $\mathcal{D}_{i, m}$;\\
Perturb the average gradient as $\boldsymbol{g}_{i}^{t, \ell} = \boldsymbol{g}_{i}^{t, \ell} + \mathcal{N}\left(0, \widehat{\sigma}^{2}C^{2}\mathbf{I}\right)/\vert\boldsymbol{b}\vert$;\\
Update the local model using $\boldsymbol{w}^{t,\ell+1}_{i}=\boldsymbol{w}^{t,\ell}_{i}-\eta\boldsymbol{g}_{i}^{t, \ell}$;\\
}
Calculate the model update by $\Delta\boldsymbol{w}_{i}^{t}=\boldsymbol{w}_{i}^{t,\tau} - \boldsymbol{w}^{t-1}$;\\
Transmit the local model update $\Delta \boldsymbol{w}^{t}_{i}$ to the server via the allocated channel;\\
Calculate the accumulative privacy budget using~\eqref{eq:accumulative_eps};\\
Send the quit notification to the server if it will exceed the PL $\epsilon_{i}$ in the next communication round;\\
}
{
The server aggregates all received local model differences by~\eqref{equ:model_aggre} and set $\boldsymbol{w}^{f}=\boldsymbol{w}^{t+1}$;\\
The server removes the client with the quit notification from the client set $\mathcal{\widehat{U}}$;\\
}
}
\end{algorithm}
\section{Gradient Sparsification empowered Differentially Private FL}
In this section, we will first propose an adaptive DP clipping threshold adjustment method to alleviate the detrimental influence of the training performance caused by DP noise.
Then, we will describe our proposed gradient sparsification empowered differentially private FL algorithm.
\subsection{Adaptive DP Clipping Threshold with Gradient Sparsification}

Before introducing our proposed DP-SparFL algorithm, we need to solve two main challenges as follows: 1) The selection of the sparsification positions during training; 2) The selection of clipping threshold $C$.
For the first challenge, we will use a random sparsifier to reduce the gradient size during training, which benefits DP because it will not consume any privacy budget.
As mentioned above, the $L_2$ norm of each gradient in the local training needs to be clipped by a preset clipping threshold $C$ before adding the random noise. We find that the gradient sparsification process will reduce the gradient norm, and then alleviate the detrimental influence caused by DP noise.
Thus, in the following lemma, we provide an adaptive gradient clipping technique with configurable gradient sparsification rates to address the second challenge.

\begin{lemma}
\label{lemma:clipping_value}
With a given gradient sparsification rate $s_{i}^{t}$ for the $i$-th client, we can adopt the clipping threshold $\sqrt{s_{i}^{t}}C$ to replace the original clipping threshold $C$.
\end{lemma}
\begin{proof}
Please see Appendix~\ref{appendix:clipping_value}.
\end{proof}
\textbf{Lemma~\ref{lemma:clipping_value}} provides a feasible method to adjust the clipping threshold instead of one-by-one searching for each client with the changeable gradient sparsification rate $s_{i}^{t}$ determined by the channel condition. 
Specifically, in order to upload local updates successfully over wireless networks, the gradient sparsification rates for clients at each communication need to be adjusted due to changeable channel conditions.
It is not practical to try various clipping thresholds to find the optimal one in this changeable scenario.
In addition, owning to the reduction of the clipping threshold, the noise variance will also be reduced, because the DP noise variance is proportional to the square of the clipping threshold.

\subsection{DP-SparFL Algorithm}
In this subsection, we propose the DP-SparFL algorithm based on \textbf{Lemma~\ref{lemma:clipping_value}} and the gradient sparsification technique to improve the local training process in~\textbf{Algorithm~\ref{alg:DP-FL}}.
The key steps of local training in the proposed DP-SparFL algorithm are described as follows:
\begin{itemize}
\item[$1)$] \textbf{Determining the sparsification positions.} Before local training, each client needs to generate a binary mask $\boldsymbol{m}_{i}^{t}$, $\forall \boldsymbol{m}_{i, k}^{t} \in \{0, 1\}, k\in \{1,\ldots,\vert\boldsymbol{m}_{i}^{t}\vert\}, \forall i \in \mathcal{S}^{t}$, with a random sparsifier for a given gradient sparsification rate $s_{i}^{t}$.
    The binary mask is utilized to prune the gradients in the local training process.
\item[$2)$] \textbf{Obtaining the sparse gradient}. In local training, SGD is adopted while avoiding calculating a fraction of gradients based on the binary mask $\boldsymbol{m}_{i}^{t}$, which can be given by
    \begin{equation}\label{eq:prune_gradient}
    \nabla F(\boldsymbol{w}^{t,\ell}_{i}, \mathcal{D}_{i, m}) = \nabla F(\boldsymbol{w}^{t,\ell}_{i}, \mathcal{D}_{i, m})\odot \boldsymbol{m}_{i}^{t}, \forall i \in \mathcal{S}^{t},
    \end{equation}
    where $\odot$ denotes the element-wise product process.
\item[$3)$] \textbf{Perturbing the gradient.} After the sparsification process, we need to clip the gradient of the $m$-th sample as
    \begin{equation}\label{eq:clip_gradient}
    \begin{aligned}
    \boldsymbol{g}_{i}^{t, \ell} &= \frac{1}{\vert \boldsymbol{b} \vert}\sum_{m \in \boldsymbol{b}} \nabla F(\boldsymbol{w}^{t,\ell}_{i}, \mathcal{D}_{i, m})\cdot\min\bigg{\{}1, \frac{\Vert \nabla F(\boldsymbol{w}^{t,\ell}_{i}, \mathcal{D}_{i, m})\Vert}{\sqrt{s_{i}^{t}}C}\bigg{\}}, \forall i \in \mathcal{S}^{t},
    \end{aligned}
    \end{equation}
    where $C$ is the preset clipping threshold for the gradient norm of each sample. Then, we need to perturb the clipped gradient to satisfy the DP guarantee, which can be given by
    \begin{equation}\label{eq:perturb_gradient}
    \boldsymbol{g}_{i}^{t, \ell} = \boldsymbol{g}_{i}^{t, \ell} + \boldsymbol{n}_{i}^{t, \ell}, \forall i \in \mathcal{S}^{t},
    \end{equation}
    where $\boldsymbol{n}_{i}^{t, \ell}$ is the DP noise vector drawn from Gaussian distribution $\mathcal{N}\left(0, \widehat{\sigma}^{2}s_{i}^{t}C^{2}\mathbf{I}\right)/\vert\boldsymbol{b}\vert$. Owing to the sparsification update, we can use a smaller clipping threshold, i.e., $\sqrt{s_{i}^{t}}C$, to replace the original one as shown in \textbf{Lemma~\ref{lemma:clipping_value}}.
\item[$4)$] \textbf{Generating the local update.} With the sparsification, the local update can be expressed as
    \begin{equation}
    \Delta \boldsymbol{w}_{i}^{t,\tau} =\boldsymbol{w}_{i}^{t,\tau}-\boldsymbol{w}_{i}^{t,0} =-\eta\sum_{\ell=0}^{\tau-1}\boldsymbol{g}_{i}^{t, \ell}\odot\boldsymbol{m}_{i}^{t}, \forall i \in \mathcal{S}^{t},
    \end{equation}
    where $\eta$ is the learning rate. Because the binary mask is unchanged via entire local training process, the local update will also be sparse.
\end{itemize}

We can observe that a large $s_{i}^{t}$ preserves more non-zero parameter updates of $\Delta \boldsymbol{w}_{i}^{t}$ and hence improves the learning performance, but it also increases the communication cost.
Although we have introduced the key steps of the proposed DP-SparFL algorithm, the resource allocation process at the server side also needs to be optimized over wireless channels.
Thus, in the following section, we will propose the communication-efficient scheduling policy design for the proposed algorithm.
\section{Communication-Efficient Scheduling Policy}\label{sec:Client_scheduling}
In this section, we will first analyze the convergence bound of DP-SparFL.
Then, we will develop a resource allocation scheme based on the convergence bound over wireless channels to improve the FL performance.
\subsection{Convergence Analysis}
Following the literature on the convergence of gradient based training algorithms, we make the following assumptions, which have been proven to be satisfied for most of ML models~\cite{Chen2020A}.
\begin{assumption}\label{ass:LossFunction}
We make assumptions on the global loss function $F(\cdot)$ defined by $F(\cdot)\triangleq \sum_{i=1}^{U}p_iF_{i}(\cdot)$, and the $i$-th local loss function $F_{i}(\cdot)$ as follows:
\begin{enumerate}
\item[\emph{1)}] $F_{i}(\boldsymbol{w})$ is $L$-Lipschitz smooth, i.e., $\Vert \nabla F_{i}(\boldsymbol{w})-\nabla F_{i}(\boldsymbol{w}')\Vert\leq L\Vert \boldsymbol{w}-\boldsymbol{w}'\Vert$, for any $\boldsymbol{w}$ and $\boldsymbol{w}'$, where $L$ is a constant determined by the practical loss function;
\item[\emph{2)}] The stochastic gradients are bounded, i.e., for any $i$ and $\boldsymbol{w}$, $\mathbbm{E}\left[\Vert\nabla F_{i}(\boldsymbol{w})\Vert^2\right]\leq G^2$;
\item[\emph{3)}] For any $i$ and $\boldsymbol{w}$, $\Vert\nabla F_{i}(\boldsymbol{w})-\nabla F(\boldsymbol{w})\Vert^2\leq \varepsilon_{i}$, where $\varepsilon_{i}$ is the divergence metric. we also define $\varepsilon \triangleq \mathbbm{E}_{i}[\varepsilon_{i}]$.
\end{enumerate}
\end{assumption}

Then, based on the above assumptions, we present the convergence results for general loss functions.
As $F(\cdot)$ may be non-convex, we study the gradient of the global model $\boldsymbol{w}^{t}$ as $t$ increases.
We give the following theorem.
\begin{theorem}
\label{theorm:convergence_result}
Let \textbf{Assumptions} 1)-3) hold, and $L$, $\eta$ and $\tau$ be as defined therein and selected properly to ensure $\eta L\tau < 1$, then the convergence bound of DP-SparFL can be given as
\begin{equation}\label{equ:conv_bound}
\begin{aligned}
\frac{1}{T}\sum_{t=0}^{T-1}\mathbbm{E}\left[\left\Vert \nabla F\left(\boldsymbol{w}^{t}\right)\right\Vert^2\right] &\leq \frac{2\left(F\left(\boldsymbol{w}^{0}\right)- F\left(\boldsymbol{w}^{T}\right)\right)}{\eta\tau T} + \underbrace{3\varepsilon}_{\text{Caused by the data divergence}}\\
&\quad+ \underbrace{\frac{3G^2}{T}\sum_{t=0}^{T-1}\sum_{i\in \mathcal{U}}p_{i}^{t}\sum_{j=1}^{N}\boldsymbol{a}_{i, j}^{t}(1-s_i^{t})}_{\text{Caused by the gradient sparsification}}+\underbrace{\eta \tau^2\Theta(1+3\eta L\tau)}_{\text{Caused by DP}},
\end{aligned}
\end{equation}
where $\Theta$ is the expectation of $L_2$ norm of the DP noise vector.
\end{theorem}
\begin{proof}
Please see Appendix~\ref{appenidx:conver_analysis}.
\end{proof}

It can be found from the right-hand side (RHS) of~\eqref{equ:conv_bound} that the convergence of DP-SparFL is affected by various parameters, including the gradient sparsification rate $s_i^{t}$, the divergence metric $\varepsilon$ and the noise variance.
This convergence bound can achieve an $O\left(\frac{1}{\eta\tau T}\right)$ convergence rate.
Moreover, we list several important insights as follows:
\begin{itemize}
\item[$\bullet$]
Firstly, when the $i$-th client adopts a small gradient sparsification rate, the convergence performance will decrease. However, in the resource limited scenario, the gradient sparsification can reduce the transmit delay.
\item[$\bullet$]
Secondly, the divergence metric $\varepsilon$  is also an important factor for the convergence performance, which reflects the data divergence for different clients. A balanced data distribution can directly accelerate the algorithm convergence.
\item[$\bullet$]
Last but not the least, the noise variance will also influence the convergence largely. It can be found that the proposed adaptive gradient clipping technique will improve the convergence performance by reducing the clipping threshold.
\end{itemize}

Hence, in the following subsection, we aim to maximize the FL training performance via trying to guarantee a high gradient sparsification rate (a large percentage of elements that have been retained) for each client under the training delay and client fairness constraints.

\subsection{Optimal Gradient Sparsification Rate and Wireless Resource Allocation}
Let us investigate the online scenario, where the channel allocation $\boldsymbol{a}_{i, j}^{t}$, transmit power $P_{i}^{t}$, and the gradient sparsification rate $s_{i}^{t}$ of each client are optimized for each communication round.
In order to improve the learning performance, we choose to minimize the convergence bound in the right-hand side (RHS) of~\eqref{equ:conv_bound}.
In addition, because of the privacy constraint, clients will have different participation rates (i.e., client fairness~\cite{Xia2020Multi,Huang2021An}).
Thus, we calculate the participant rate based on~\eqref{eq:accumulative_eps} for the $i$-th client as
\begin{equation}
\beta_{i} = \min\left\{\frac{N\widehat{T}_{i}}{\sum_{i'=1}^{U}\widehat{T}_{i'}}, 1\right\},
\end{equation}
where
\begin{equation}
\widehat{T}_{i} = \left\lfloor\frac{(\alpha-1)\epsilon_{i} - \log(\frac{1}{\delta})-(\alpha-1)\log(1-\frac{1}{\alpha})+\log(\alpha)}{\tau\ln {\mathbb{E}_{z\sim \mu_{0}(z)}\left[\left(1-q_{i}+\frac{q_{i} \mu_{1}(z)}{\mu_{0}(z)}\right)^{\alpha}\right]}}\right\rfloor,
\end{equation}
$\widehat{T}_{i}$ is the number of communication rounds that the $i$-th client can participate in under the DP constraint $\epsilon_{i}$ using noise STD $\widehat{\sigma}$ and $\lfloor\cdot\rfloor$ denotes rounding down.
The delay of each communication round is determined by the slowest client, and given by
$d^{t}=\max_{i\in\mathcal{U}} \sum_{j\in \mathcal{N}}\boldsymbol{a}_{i,j}^{t} d_{i, j}^{t}$
and $d_{i, j}^{t} = d_{i}^{t, \text{do}} + d_{i}^{t, \text{lo}} + d_{i, j}^{t, \text{up}}$.
By maximizing the FL training performance with the training delay and client fairness constraints, we formulate the following resource allocation problem:
\begin{equation}\label{equ:optimal_obj}
\begin{aligned}
&\textbf{P1:}\quad\min_{\substack{P_{i}^{t}, s_{i}^{t},
\boldsymbol{a}_{i, j}^{t},\\ i=1,\ldots,U}}-\frac{1}{T}\sum_{t=0}^{T-1}\sum_{i\in \mathcal{U}}\sum_{j\in \mathcal{N}}\boldsymbol{a}_{i,j}^{t}p_{i}^{t}s_{i}^{t}\notag\\
\text{s.t.}\,&\textbf{C1:}~\boldsymbol{a}_{i,j}^{t} \in \{0,1\},
\qquad\quad\textbf{C2:}~\sum_{j\in \mathcal{N}}\boldsymbol{a}_{i,j}^{t}=1,\qquad\,\textbf{C3:}~\sum_{i\in \mathcal{U}}\boldsymbol{a}_{i,j}^{t}\leq1,\notag\\
&\textbf{C4:}~s^{\text{th}}\leq s_i^{t}\leq 1,\qquad\quad\,\textbf{C5:}~0<P_{i}^{t}\leq P^\text{max},\quad\,\textbf{C6:}~\sum_{j\in \mathcal{N}}\boldsymbol{a}_{i,j}^{t}E_{i, j}^{t, \text{co}}+E_{i}^{t, \text{cp}}\leq E^{\text{max}},\notag\\
&\textbf{C7:}~\frac{1}{T}\sum_{t=0}^{T}\sum_{j\in \mathcal{N}}\boldsymbol{a}_{i,j}^{t} \geq \beta_{i},\,\textbf{C8:}~\frac{1}{T}\sum_{t=0}^{T}d^{t}\leq d^{\text{Avg}},\notag
\end{aligned}
\end{equation}
where $P^{\text{max}}$ is the maximum available transmit power, $E^{\text{max}}$ is the maximum available energy, $d^{\text{Avg}}$ is the required average training latency
and $s^{\text{th}}$ is to restrict the gradient sparsification rate since the learning accuracy decreases sharply when the gradient sparsification rate is very low~\cite{Pavlo2019Importance}.
The constraints~\textbf{C1-C3} represent
the wireless channel allocation and
\textbf{C4} is to restrict the gradient sparsification rate.
\textbf{C5} and \textbf{C6} represent the transmit power constraint and the total power consuming, respectively.
\textbf{C7} and \textbf{C8} are adopted to restrict the client fairness and learning latency, respectively.
Obviously, the problem is a mixed integer optimization problem, which is non-convex in general and cannot be directly solved. Therefore, in the following, we first decompose the original problem, and then develop a low computational complexity algorithm to achieve the suboptimal solutions.

We can observe that \textbf{P1} is a stochastic optimization problem with long-term constraints \textbf{C7} and \textbf{C8}, which can be transformed into part of the objective function~\cite{Neely2010Stochastic}.
First, we leverage the Lyapunov technique to transform constraint $\textbf{C7}$ and \textbf{C8} into queue stability constraints~\cite{Kang2018Low}.
In detail, we define virtual queues $Q^{t, \text{fa}}_{i}$ and $Q^{t, \text{de}}$ with the following update equations:
$Q^{t, \text{fa}}_{i}\triangleq[Q^{t-1, \text{fa}}_{i}+\sum_{j=1}^{N}\boldsymbol{a}_{i,j}^{t}-\beta_{i}]^{+}$
and
$Q^{t, \text{de}} \triangleq [Q^{t-1, \text{de}}+d^{t}-d^{\text{Avg}}]^{+}$, respectively,
where $[x]^{+} \triangleq \max\{x, 0\}$.
Under the framework of Lyapunov optimization, we further resort to the drift-plus-penalty algorithm and obtain 
\begin{equation}\label{equ:optimal_obj}
\begin{aligned}
\textbf{P2:}\quad&\min_{\boldsymbol{P}^{t}, \boldsymbol{s}^{t},
\boldsymbol{a}^{t}}\frac{1}{T}\sum_{t=0}^{T-1}V^{t}(\boldsymbol{P}^{t}, \boldsymbol{s}^{t},
\boldsymbol{a}^{t})\notag\\
\text{s.t.}\quad&\textbf{C1:}~\boldsymbol{a}_{i,j}^{t} \in \{0,1\},
\quad\textbf{C2:}~\sum_{j\in \mathcal{N}}\boldsymbol{a}_{i,j}^{t}=1,\qquad\textbf{C3:}~\sum_{i\in \mathcal{U}}\boldsymbol{a}_{i,j}^{t}\leq1,\notag\\
&\textbf{C4:}~s^{\text{th}}\leq s_i^{t}\leq 1,\quad\textbf{C5:}~0<P_{i}^{t}\leq P^\text{max},
\quad\textbf{C6:}~\sum_{j\in \mathcal{N}}\boldsymbol{a}_{i,j}^{t}E_{i, j}^{t, \text{co}}+E_{i}^{t, \text{cp}}\leq E^{\text{max}},\notag\\
\end{aligned}
\end{equation}
where
\begin{equation}
\begin{aligned}
V^t(\boldsymbol{P}^{t}, \boldsymbol{s}^{t}, \boldsymbol{a}^{t}) &= \sum_{i\in \mathcal{U}}\sum_{j\in \mathcal{N}}(Q_{i}^{t, \text{fa}}-\lambda p_{i}^{t}s_{i}^{t})\boldsymbol{a}_{i,j}^{t}+Q^{t, \text{de}}\left(d^{t}-d^{\text{Avg}}\right)-Q_{i}^{t, \text{fa}}\sum_{i\in \mathcal{U}}\beta_{i},
\end{aligned}
\end{equation}
$\boldsymbol{P}^{t}=\{P_{i}^{t}|i\in\mathcal{U}\}$, $\boldsymbol{s}^{t}=\{s_{i}^{t}|i\in\mathcal{U}\}$, and $\lambda > 0$ is a tuneable parameter that controls the trade-off between minimizing the convergence bound and the training delay.
We can observe that~\textbf{P2} can be divided into $T$ independent sub-optimization problems, for each communication round.
\subsubsection{Optimal Gradient Sparsification Rate}
The optimal gradient sparsification rate of the $i$-th client can be determined by the following theorem.
\begin{theorem}\label{pro:opti_pruning}
Given the client scheduling vector $\boldsymbol{a}^{t}$ and transmit power $\boldsymbol{P}^{t}$, the solution of the optimal gradient sparsification rate vector at the $t$-th communication round can be divided into $N$ subproblems with closed-form solutions, and then solved as
\begin{equation}\label{eq:opti_pruning}
\begin{aligned}
\boldsymbol{s}^{t,\star} = \mathop{\arg\min}_{\boldsymbol{s}^{t}\in \mathcal{K}^{t}} V^t\left(\boldsymbol{P}^{t}, \boldsymbol{s}^{t}, \boldsymbol{a}^{t}\right),
\end{aligned}
\end{equation}
where $\mathcal{K}^{t}$ is the set of all available solutions for $N$ subproblems.
\end{theorem}
\begin{proof}
Please see Appendix~\ref{appendix:optimal_pruning}.
\end{proof}

From \textbf{Theorem~\ref{pro:opti_pruning}}, we can observe that the optimization problem for the gradient sparsification rate vector can be divided into $N$ subproblems, and each of them can be solved with the closed-form solution.
We can address all $N$ subproblems, and then select the optimal solution as the final output.
Therefore, the complexity of optimizing the optimal gradient sparsification rate is linear with $N$.

\subsubsection{Optimal Transmit Power}
To obtain the optimal transmit power, we first derive the relation between $P_{i}^{t}$ and $V^t(\boldsymbol{P}^{t}, \boldsymbol{s}^{t}, \boldsymbol{a}^{t})$ as
\begin{equation}\label{equ:longequ2}
\begin{aligned}
&\frac{\partial V^t(\boldsymbol{P}^{t}, \boldsymbol{s}^{t}, \boldsymbol{a}^{t})}{\partial P_{i}^{t}} =
\begin{cases}
\sum\limits_{j\in \mathcal{N}}\frac{ -(\ln2)\lambda Zp_{i}^{t}s_{i}^{t}h_{i, j}^{t}\boldsymbol{a}_{i,j}^{t}}{B\sigma^{2}\left(\log_{2}\left(1+\frac{P_{i}^{t} h_{i, j}^{t}}{\sigma^{2}}\right)\right)^2\left(1+\frac{P_{i}^{t} h_{i, j}^{t}}{\sigma^{2}}\right)}, & \text{if $\sum\limits_{j\in \mathcal{N}}\boldsymbol{a}_{i,j}^{t}d_{i, j}^{t} > \sum\limits_{j\in \mathcal{N}}\boldsymbol{a}_{i,j}^{t}d_{i', j}^{t}\forall i'\in \mathcal{S}^{t}/i$,} \\
0,&  \text{else.}
\end{cases}
\end{aligned}
\end{equation}
Then, we can derive the first derivative between $E_{i}^{t, \text{co}}$ and $P_{i}^{t}$ as
\begin{equation}\label{equ:longequ3}
\begin{aligned}
\frac{\partial E_{i}^{t, \text{co}}}{\partial P_{i}^{t}} &= \sum\limits_{j\in \mathcal{N}}\frac{ -(\ln2)P_{i}^{t} Zp_{i}^{t}s_{i}^{t}h_{i, j}^{t}\boldsymbol{a}_{i,j}^{t}}{B\sigma^{2}\left(\log_{2}\left(1+\frac{P_{i}^{t} h_{i, j}^{t}}{\sigma^{2}}\right)\right)^2\left(1+\frac{P_{i}^{t} h_{i, j}^{t}}{\sigma^{2}}\right)}+\frac{ Zp_{i}^{t}s_{i}^{t}\boldsymbol{a}_{i,j}^{t}}{B\log_{2}\left(1+\frac{P_{i}^{t} h_{i, j}^{t}}{\sigma^{2}}\right)}\\
&=\sum\limits_{j\in \mathcal{N}}\frac{ \left(Zp_{i}^{t}s_{i}^{t}\left(\sigma^{2}\left(1+\frac{P_{i}^{t} h_{i, j}^{t}}{\sigma^{2}}\right)\log_{2}\left(1+\frac{P_{i}^{t} h_{i, j}^{t}}{\sigma^{2}}\right)-(\ln2)P_{i}^{t} h_{i, j}^{t}\right)\right)\boldsymbol{a}_{i,j}^{t}}{B\sigma^{2}\left(\log_{2}\left(1+\frac{P_{i}^{t} h_{i, j}^{t}}{\sigma^{2}}\right)\right)^2\left(1+\frac{P_{i}^{t} h_{i, j}^{t}}{\sigma^{2}}\right)}>0.
\end{aligned}
\end{equation}
Therefore, the largest transmit power under the constraint can be applied, i.e.,
\begin{equation}\label{eq:opti_transmit_power}
\begin{aligned}
P_{i}^{t,\star} = \max\{P^{\text{max}}, P^{t,\text{th}}\},
\end{aligned}
\end{equation}
where $P^{\text{th}}$ satisfies the equation:
\begin{equation}
\sum\limits_{j\in \mathcal{N}}\frac{P_{i}^{t,\text{th}}p_{i}^{t}s_{i}^{t}Z\boldsymbol{a}_{i,j}^{t}}{B\log_{2}\left(1+\frac{P_{i}^{t,\text{th}} h_{i, j}^{t}}{\sigma^{2}}\right)} = E^{\text{max}}-E_{i}^{t, \text{cp}}.
\end{equation}
\subsubsection{Optimal Channel Allocation}
When $Q^{t, \text{de}}> 0$, based on the optimizing process above, we can simplify the optimization problem~\textbf{P2} as follows.
\begin{equation}\label{equ:sim_optimal_obj}
\begin{aligned}
\textbf{P2:}\,&\min_{\boldsymbol{a}^{t}}\sum_{i\in \mathcal{U}}\sum_{j\in \mathcal{N}}(Q_{i}^{t, \text{fa}}-\lambda p_{i}^{t}s_{i}^{t})\boldsymbol{a}_{i,j}^{t}\\
&+Q^{t, \text{de}}\max_{i\in\mathcal{N}}\Bigg{\{}\sum_{j=1}^{N}\boldsymbol{a}_{i,j}^{t}\Bigg{(}\frac{ Zp_{i}^{t}s_{i}^{t}}{B\log_{2}\Big{(}1+\frac{P_{i}^{t} h_{i,j}^{t}}{\sigma^{2}}\Big{)}}+d_i^{t, \text{do}}+\frac{\tau \vert \mathcal D_{i}\vert \Phi_{i}}{f_{i}^{t}}\Bigg{)}\Bigg{\}}\\
\text{s.t.}\,\,&\textbf{C1:}~\boldsymbol{a}_{i,j}^{t} \in \{0,1\},
\,\textbf{C2:}~\sum_{j\in \mathcal{N}}\boldsymbol{a}_{i,j}^{t}=1,\,\textbf{C3:}~\sum_{i\in \mathcal{U}}\boldsymbol{a}_{i,j}^{t}\leq1,\, \textbf{C6:}~\sum_{j\in \mathcal{N}}\boldsymbol{a}_{i,j}^{t}E_{i, j}^{t, \text{co}}+E_{i}^{t, \text{cp}}\leq E^{\text{max}}.
\end{aligned}
\end{equation}
Obviously, the objective function~\eqref{equ:sim_optimal_obj} is a mix-optimization problem, and the optimization variables are integers.
We can use $\mu$ to replace the second term in \textbf{P2}, problem~\eqref{equ:sim_optimal_obj} can be transformed as:
\begin{equation}
\begin{aligned}
\textbf{P3:}\quad&\min_{\substack{\mu, \boldsymbol{a}^{t}}}\sum_{i\in \mathcal{U}}\sum_{j\in \mathcal{N}}(Q_{i}^{t, \text{fa}}-\lambda p_{i}^{t} s_{i}^{t})\boldsymbol{a}_{i,j}^{t} + \mu\\
\text{s.t.}\quad&\textbf{C1-C3}\,,\textbf{C6},\\
&\textbf{C9:}~\mu\geq\max_{i\in\mathcal{U}}\Bigg{\{}Q^{t, \text{de}}\sum_{j=1}^{N}\boldsymbol{a}_{i,j}^{t}\Bigg{(}\frac{ Zp_{i}^{t}s_{i}^{t}}{B\log_{2}\left(1+\frac{P_{i}^{t} h_{i, j}^{t}}{\sigma^{2}}\right)}+d_{i}^{\text{do}}+\frac{\tau \vert \mathcal D_{i}\vert \Phi_{i}}{f_{i}^{t}}\Bigg{)}\Bigg{\}}.
\end{aligned}
\end{equation}
Using this form, we can iteratively solve two variables $\mu$ and $\boldsymbol{a}_{i, j}^{t}$, which are corresponding to two subproblems, i.e.,
\begin{equation}
\begin{aligned}
\textbf{P31:}\quad\min_{\substack{\mu}}&\,\mu \quad \text{s.t.} \,\, &\textbf{C9}.
\end{aligned}
\end{equation}
and
\begin{equation}\label{equ:channel_matching}
\begin{aligned}
\textbf{P32:}\quad\min_{\substack{\boldsymbol{a}^{t}}}&\sum_{i\in \mathcal{U}}\sum_{j\in \mathcal{N}}(Q_{i}^{t, \text{fa}}-\lambda p_{i}^{t} s_{i}^{t})\boldsymbol{a}_{i,j}^{t}\\
\quad\text{s.t.}\quad&\textbf{C1-C3},\textbf{C6}\,\, \text{and}\,\, \textbf{C9}.
\end{aligned}
\end{equation}
Further, the optimal solution of \textbf{P31} can be expressed as
\begin{equation}\label{equ:partial_solution}
\begin{aligned}
\mu^{\star} &= \max_{i\in\mathcal{U}}\Bigg{\{}Q^{t, \text{de}}\sum_{j\in \mathcal{N}}\boldsymbol{a}_{i,j}^{t}\Bigg{(}\frac{ Zp_{i}^{t}s_{i}^{t}}{B\log_{2}\left(1+\frac{P_{i}^{t} h_{i, j}^{t}}{\sigma^{2}}\right)}+d_{i}^{\text{do}}+\frac{\tau \vert \mathcal D_{i}\vert \Phi_{i}}{f_{i}^{t}}\Bigg{)}\Bigg{\}}.
\end{aligned}
\end{equation}
The problem \textbf{P32} can be addressed by using a bipartite matching algorithm, which aims to maximize the objective function.
In order to use a bipartite matching algorithm for solving problem \textbf{P32}, we first transform the optimization problem into a bipartite matching problem.
We construct a bipartite graph $\mathcal{A} = (\mathcal{U}\times \mathcal{N}, \mathcal{E})$ where $\mathcal{N}$ is the set of channels that can be allocated to each client, each vertex in $\mathcal{U}$ represents a client and each vertex in $\mathcal{N}$ represents an channel, and $\mathcal{E}$ is the set of edges that connect to the vertices from each set $\mathcal{U}$ and $\mathcal{N}$.
Let $\boldsymbol{a}_{i, j} \in \mathcal{E}$ be the edge connecting vertex $i$ in $\mathcal{U}$ and vertex
$j$ in $\mathcal{N}$ with $\boldsymbol{a}_{i, j} \in \{0, 1\}$, where $\boldsymbol{a}_{i, j} = 1$ indicates that the $j$-th channel is allocated to the $i$-th client, otherwise, we have $\boldsymbol{a}_{i, j} = 0$.
We aim to find a subset of edges in $\mathcal{E}$, in which no two edges share a common vertex in $\mathcal{N}$, such that each channel can only be allocated to one client.
Here, we first prune the edges that cannot satisfy the constraint \textbf{C9} and obtain a non full-connected bipartite graph.
Then, we can use the conventional Hungarian search method~\cite{Mohammad2011Online} to solve the optimal matching result $\boldsymbol{a}^{*}$.
When $Q^{t, \text{de}} \leq 0$, we can note that the optimization problem can be solved by the conventional Hungarian search method.

Overall, we can optimize the gradient sparsification rate $\boldsymbol{s}^{t}$, transmit power vector $\boldsymbol{P}^{t}$, and channel allocation matrix $\boldsymbol{a}^{t}$ for training the FL algorithm successively, until the decrement between two adjacent objective function is smaller than a preset value.
\subsection{Feasibility Analysis}
In this subsection, we show the proposed solution for \textbf{P1} can satisfy the client fairness and the required average delay, and the virtual queue system defined is mean rate stable.
We state the feasibility analysis as follows.
\begin{theorem}\label{theorem:stable}
The proposed solution for \textbf{P1} is feasible. Specifically, for any minimum selection fraction and required average delay, the virtual queue system defined is strongly stable.
\end{theorem}
\begin{IEEEproof}
See Appendix~\ref{appendix:feasibility_analysis}.
\end{IEEEproof}

According to \textbf{Theorem~\ref{theorem:stable}}, the long-term fairness and required average delay
constraints can be satisfied with the proposed solution as long as the requirement is feasible.
Specifically, this theorem also reveals the fact that the long-term client fairness constraint holds under any setting of
the tuneable parameter $\lambda$.
With a larger value of $\lambda$, the fairness and delay queues will have a slower convergence rate, indicating that the fairness and average delay could not be well guaranteed before convergence.
When the training rounds are finite, the number of rounds that need to undergo before convergence could compromise the fairness and average delay in some degree.
We will verify our analysis via experiment results in Section~\ref{sec:Exm_res}.

\section{Experiments}\label{sec:Exm_res}
\subsection{Setup}
We evaluate the effectiveness of the proposed algorithm on two different convolutional neural networks (CNNs) and three datasets as follows:
\begin{itemize}
 \item[$\bullet$] \textbf{CNN on MNIST and FashionMNIST.} MNIST dataset~\cite{lecun1998gradient} consists of handwritten digits with $10$ categories formatted as $28\times28$ size gray scale images. There are $60000$ training examples and $10000$ testing examples. FashionMNIST~\cite{Xiao2017Fashion} is a dataset of Zalando's article images formatted as 28$\times$28 size gray scale images, which consists of $60000$ training samples and $10000$ testing samples.
     The CNN model for MNIST and FashionMNIST datasets contains two $5 \times 5$ convolution layers (the first layer consists of $32$ filters, the second layer consists of $64$ filters, each layer is followed with $2 \times 2$ max pooling and ReLu activation), a fully connected layer with $512$ units followed with ReLu activation, and a final softmax output layer.
 \item[$\bullet$] \textbf{CNN on CIFAR-$10$.} The CIFAR-$10$ dataset~\cite{Krizhevsky2009Learning} contains $60000$ color images with $10$ object classes e.g., dog, deer, airplane and so on, where each class has $6,000$ images. This dataset has been pre-divided into two parts, i.e., the training dataset ($50000$ images) and test dataset ($10000$ images). 
 The CNN model used for CIFAR-$10$ contains three $3 \times 3$ convolution layers, two fully connected layers, and a final softmax output layer. The first, second and third convolution layers consist of $64$, $128$ and $256$ filters, respectively, where each layer is followed by $2 \times 2$ max pooling and ReLu activation.
 The first and second fully connected layers contain $128$ and $256$ units, respectively, in which each layer is followed by ReLu activation.
 \end{itemize}

To evaluate the performance, we compare the proposed DP-SparFL algorithm with the following baselines:
\begin{itemize}
\item[$\bullet$] \textbf{Random Scheduling~\cite{Yang2020Scheduling}:} The AP selects $N$ associated clients and assigns dedicated channels for them randomly at each communication round. Then, these selected clients perform local training and upload their local updates with the assigned channels.
\item[$\bullet$] \textbf{Round Robin~\cite{Yang2020Scheduling}:} The AP arranges all the clients into $\lceil\frac{U}{N}\rceil$ groups, where $\lceil \cdot \rceil$ denotes rounding up to an integer, and then assigns each group to access the channels consecutively.
\item[$\bullet$] \textbf{Delay-minimization~\cite{Chen2021Communication}:} The AP selects a client set, consisting of $N$ clients, from available clients with the minimizing training delay without update sparsification.
\end{itemize}

For these classification tasks, we use the cross-entropy loss function.
Unless otherwise stated, the system parameters are set as follows.
We set the number of clients to $20$, each containing $1000$ training and $500$ testing examples.
For the non-IID setting, we simulate a heterogeneous partition by simulating a Dirichlet distribution, i.e., $\textbf{Dir}(0.2)$, over all classes for each client~\cite{Yurochkin2019Bayesian}.
For the condition of imbalance data numbers, we divide $20$ client into four parts, and they have $300$, $600$, $1800$ and $2100$ samples for each client, respectively.
When training with FL, we set the learning rate $\eta$ as $0.002$ for the MNIST, FashionMNIST and CIFAR-$10$ datasets.
The number of local iterations is set to $60$.
Besides, simulations are performed in a square area of $100 \times 100$ $m^{2}$.
Both available channels and clients are uniformly distributed in this plane.
The number of available channels is set to $5$.
We set the bandwidth $B$ to $15$ KHz, the transmission power of the downlink channel to $23$ dBm, the maximum transmission power of the client to $30$dBm, Gaussian white noise power to $-107$ dBm, and path loss exponent model to $PLE$, in which $PLE[dB] = 128.1 + 37.6 \log(\chi)$ with $\chi$ representing the distance in km.
The maximum computing capability $f_{i}^{t}$ and the maximum power constraint of each client are set to $2.4$ GHz and $200$mW, respectively.
In addition, the required PL $\epsilon_{i}$ of the $i$-th client is randomly distributed between $2.0$ and $10.0$ similar to~\cite{Liu2021Projected}, and unchanged throughout the FL training, as well as $\delta = 0.001$, $\forall i \in \mathcal{U}$.
We use different noise STDs for different datasets, i.e., $\widehat{\sigma}_i = 0.6$, $\widehat{\sigma}_i = 0.5$ and $\widehat{\sigma}_i = 0.4$ for MNIST, FashionMNIST and CIFAR-$10$, respectively.
we utilize the method in~\cite{Abadi2016Deep} and choose $C$ by taking the median of the norms of the unclipped parameters over the course of training.
However, the gradient sparsification rates for clients' updates at each communication need to be adjusted to address the changeable channel conditions.
We adopt the adaptive DP clipping threshold technique in \textbf{Lemma~\ref{lemma:clipping_value}} to address this challenge.
To balance the learning performance and training delay, we adopt $\lambda = 50$.
\subsection{Impact of the Clipping Threshold}

\begin{figure}[ht]
\centering
\includegraphics[width=4.5in,angle=0]{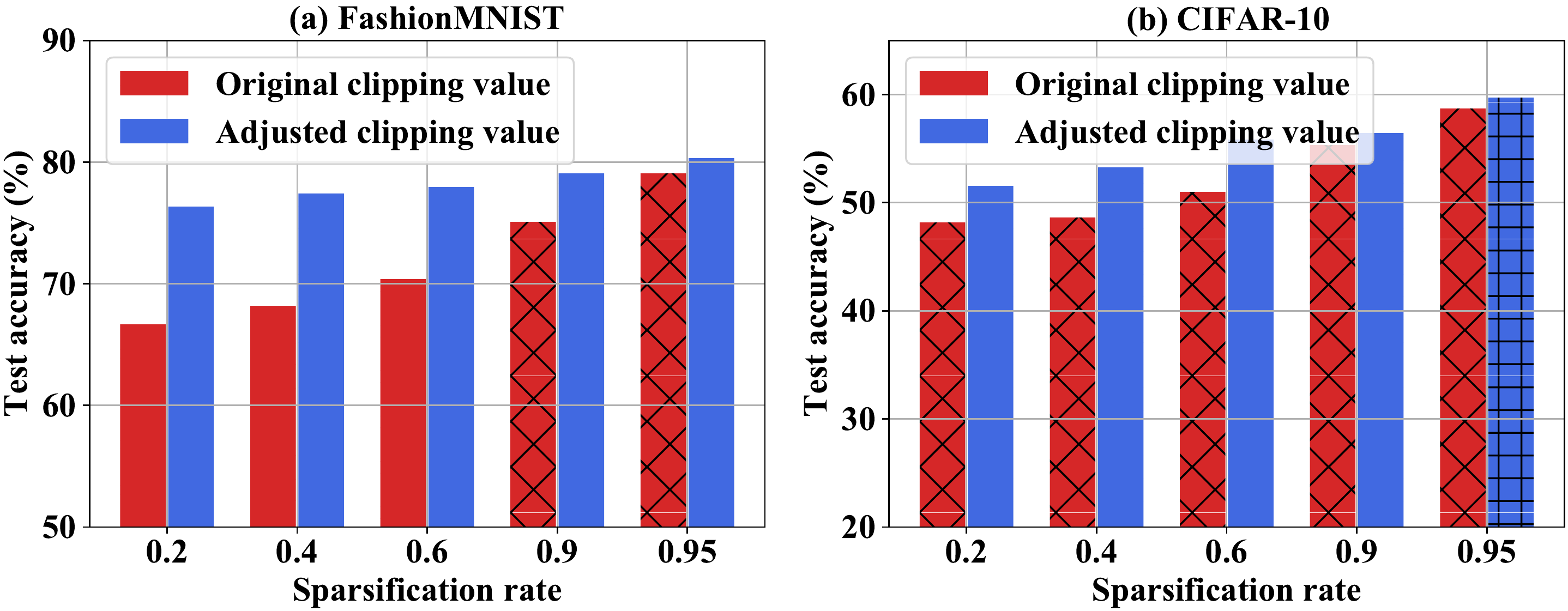}
\caption{The test accuracy of original clipping and adjusted clipping methods (\textbf{Lemma~\ref{lemma:clipping_value}}) on FashionMNIST and CIFAR-$10$ datasets with various sparsification rates (percentages of elements that have been retained).}
\label{fig:acc_adjusted_clipping}
\end{figure}

In order to evaluate the adjusting method in~\textbf{Lemma~\ref{lemma:clipping_value}}, we show the test accuracy of the proposed DP-SparFL algorithm under the unadjusted and adjusted clipping thresholds on the FashionMNIST and CIFAR-$10$ datasets in Fig.~\ref{fig:acc_adjusted_clipping}.
The original clipping threshold is the median value obtained by the pretraining process~\cite{Abadi2016Deep}.
We can note that in our system, the gradient sparsification rate for each client is uncertain and changed with the channel condition, thus it is unavailable to use the pretraining method at each communication round.
We can see that under various gradient sparsification rates, the adjusted method outperforms the unadjusted one on both FashionMNIST and CIFAR-$10$ datasets.
The intuition is that a smaller gradient sparsification rate could lead to a smaller $L_2$ norm for the training gradient for each client, thus a smaller clipping value can reduce the noise variance and improve the learning performance.
\subsection{Impact of the Tuneable Parameter $\lambda$ and Privacy Level}

Fig.~\ref{fig:acc_delay_various_V} illustrates the test accuracy and cumulative delay for the proposed DP-SparFL algorithm with various values of $\lambda$ on the FashionMNIST dataset.
From Fig.~\ref{fig:acc_delay_various_V} (a), we can see that there is an optimal $\lambda$ in view of the test accuracy.
In Fig.~\ref{fig:acc_delay_various_V} (b), we can observe that a larger $\lambda$ will have a larger training delay.
The rationale behind this is that a larger $\lambda$ can lead to a higher consideration for the update sparsification but a smaller consideration for the training delay.
The phenomenon inspires us to find a better trade-off between the training latency and update sparsification to achieve optimal training performance.

\begin{figure}[ht]
\centering
\includegraphics[width=4.5in,angle=0]{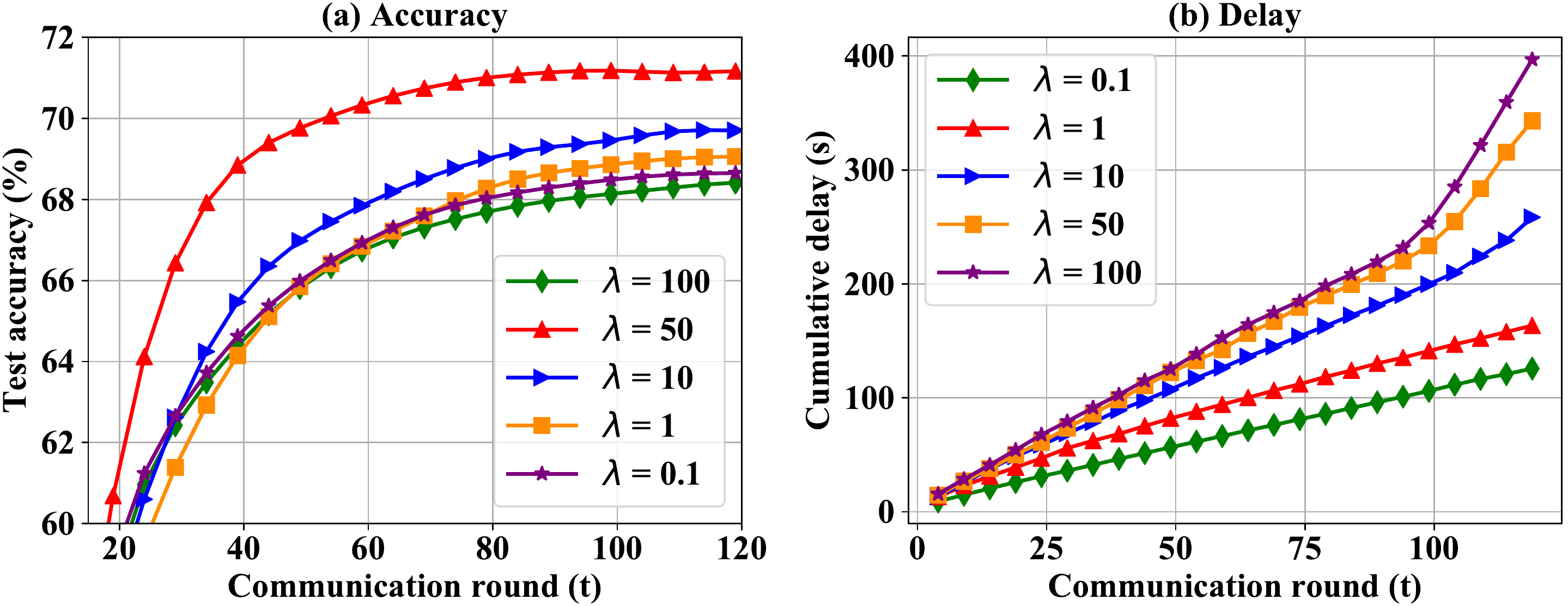}
\caption{The test accuracy and cumulative delay for the proposed DP-SparFL algorithm with various values of $\lambda$ in the IID setting on the FashionMNIST dataset.}
\label{fig:acc_delay_various_V}
\end{figure}

In Fig.~\ref{fig:acc_delay_various_eps}, we show the test accuracy and cumulative delay for the proposed DP-SparFL algorithm with various values of average privacy levels on the FashionMNIST dataset.
We can see that the test accuracy increases with the PL.
In addition, a larger PL allows clients to participate in FL training more, i.e., a larger number of communication rounds for training.
In Fig.~\ref{fig:acc_delay_various_eps} (b), it can be noted the cumulative delay decreases with the average PL.
The intuition is that clients with larger PLs can be scheduled flexibly with few privacy concerns.

\begin{figure}[ht]
\centering
\includegraphics[width=4.5in,angle=0]{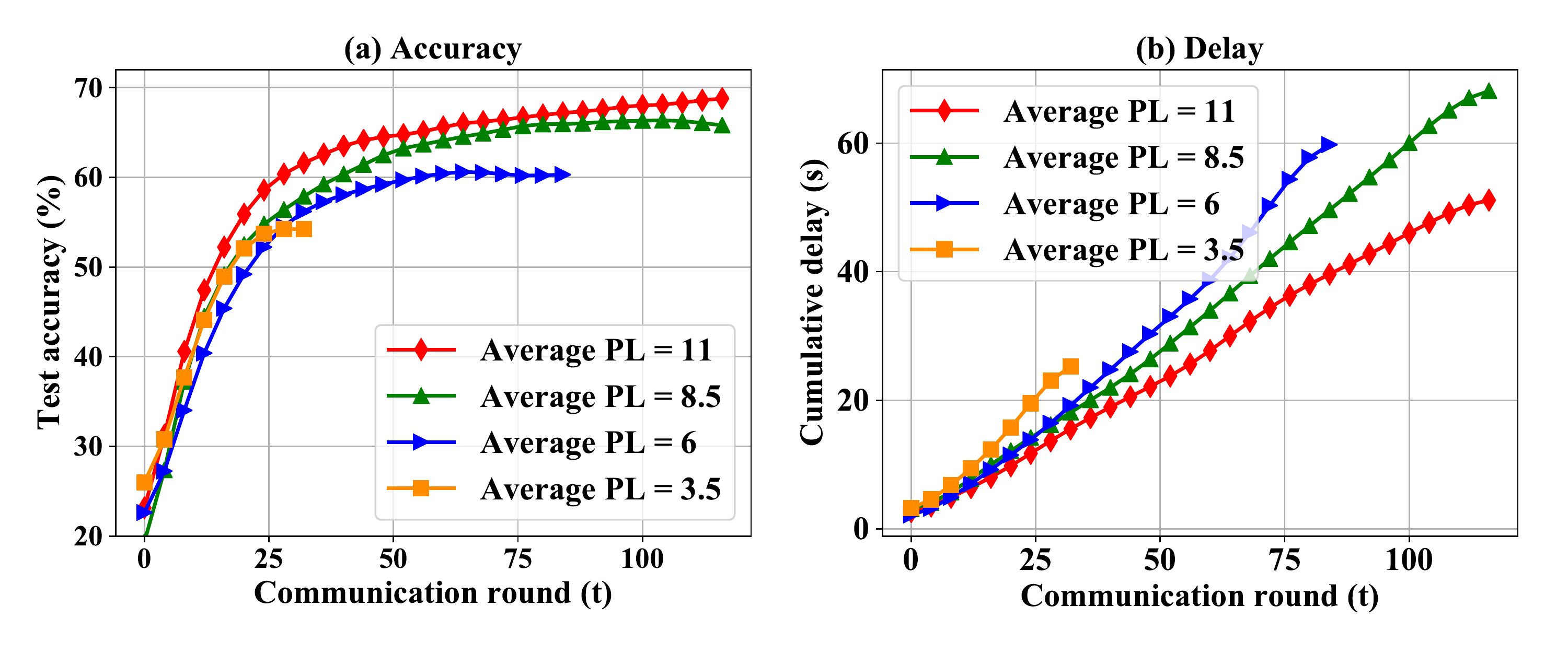}
\caption{The test accuracy and cumulative delay for the proposed DP-SparFL algorithm with various average privacy levels (PLs) on the FashionMNIST dataset. Specifically, the required DP levels of clients are randomly sampled from four intervals, i.e., $[2, 20]$, $[2, 15]$, $[2, 10]$ and $[2, 5]$.}
\label{fig:acc_delay_various_eps}
\end{figure}
\subsection{Comparison of the Various Policies}

Fig.~\ref{fig:acc_various_algorithms} illustrates the test accuracy of the proposed DP-SparFL ($\lambda = 50$), delay-minimization, round robin and random algorithms in the IID data setting on the MNIST and CIFAR-$10$ datasets.
As seen from Fig.~\ref{fig:acc_various_algorithms}(a) and (b), the proposed DP-SparFL algorithm achieves better test accuracy than baselines.
The reason is that the proposed DP-SparFL algorithm can guarantee more accessible clients owing to the sparsification technique, thereby improving the learning performance.
Moreover, the value of $\lambda$ is a key factor to balance the trade-off between the latency and the learning performance.

\begin{figure}[ht]
\centering
\includegraphics[width=4.5in,angle=0]{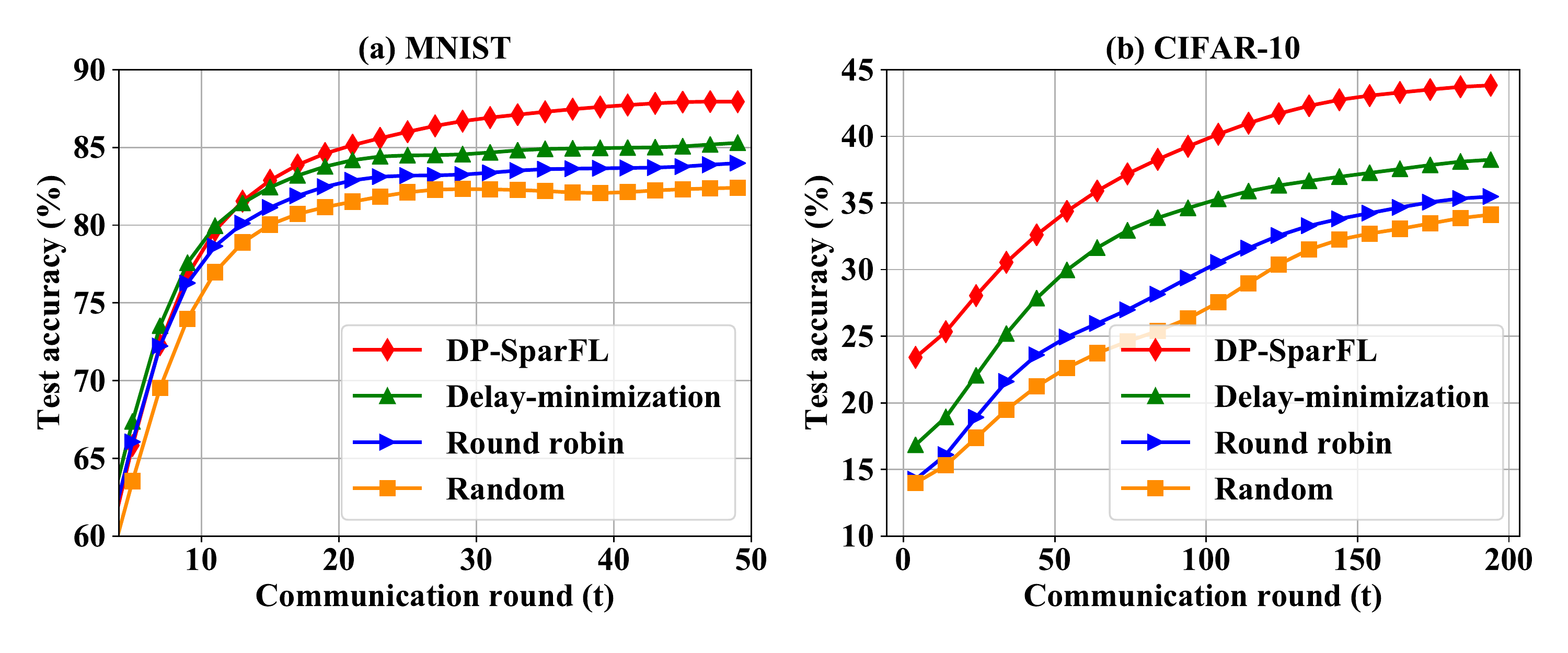}
\caption{The test accuracy for various algorithms in the IID data setting on the MNIST and CIFAR-$10$ datasets, i.e., DP-SparFL, delay-minimization, round robin and random algorithms.}
\label{fig:acc_various_algorithms}
\end{figure}

In Fig.~\ref{fig:delay_various_algorithms}, we show cumulative latencies for various algorithms in the IID data setting on the MNIST and CIFAR-$10$ datasets, i.e., DP-SparFL, delay-minimization, round robin and random algorithms.
As shown in Fig.~\ref{fig:delay_various_algorithms}, the proposed DP-SparFL algorithm can obtain the lowest latency among whole algorithms.
It can be noted that the delay of the delay-minimization algorithm is low in the early stages of the training, but it increases tremendously in the later stages.
The intuition is that the delay-minimization algorithm will select the clients with small distances with a high probability.
When the privacy budgets of these clients have run out, the cumulative latency will increase tremendously.
The proposed DP-SparFL algorithm can schedule clients with large distances with the help of sparsification, while taking account of client fairness.

\begin{figure}[ht]
\centering
\includegraphics[width=4.5in,angle=0]{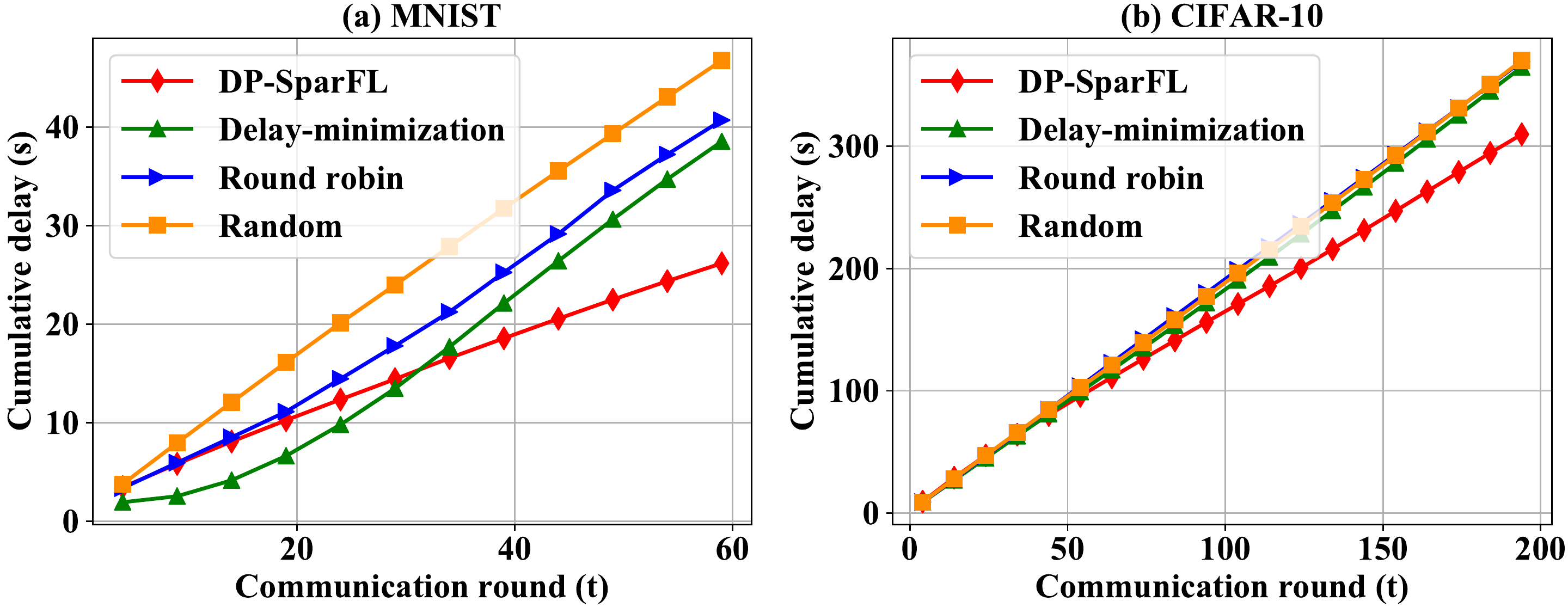}
\caption{The cumulative delay for various algorithms in the IID data setting on the MNIST and CIFAR-$10$ datasets, i.e., DP-SparFL, delay-minimization, round robin and random algorithms.}
\label{fig:delay_various_algorithms}
\end{figure}

In Figs.~\ref{fig:acc_delay_noniid_various_algorithms} and~\ref{fig:acc_delay_imbalance_various_algorithms}, we show the test accuracy and cumulative delay under two different data conditions, i.e., non-IID and imbalance sample numbers, respectively.
We can see that the proposed algorithm can obtain a higher test accuracy due to the fairness guarantee that is beneficial to non-IID data.
The reason is that in the non-IID data setting, the proposed algorithm ensures that all clients can participate in FL training with the client fairness constraint.
In addition, the proposed algorithm also takes account of the size of the dataset (affecting the sampling rate $q_{i}$ in the local training) when designing the client fairness, thus it can allocate clients that contain large datasets to high-quality channels.

\begin{figure}[ht]
\centering
\includegraphics[width=4.5in,angle=0]{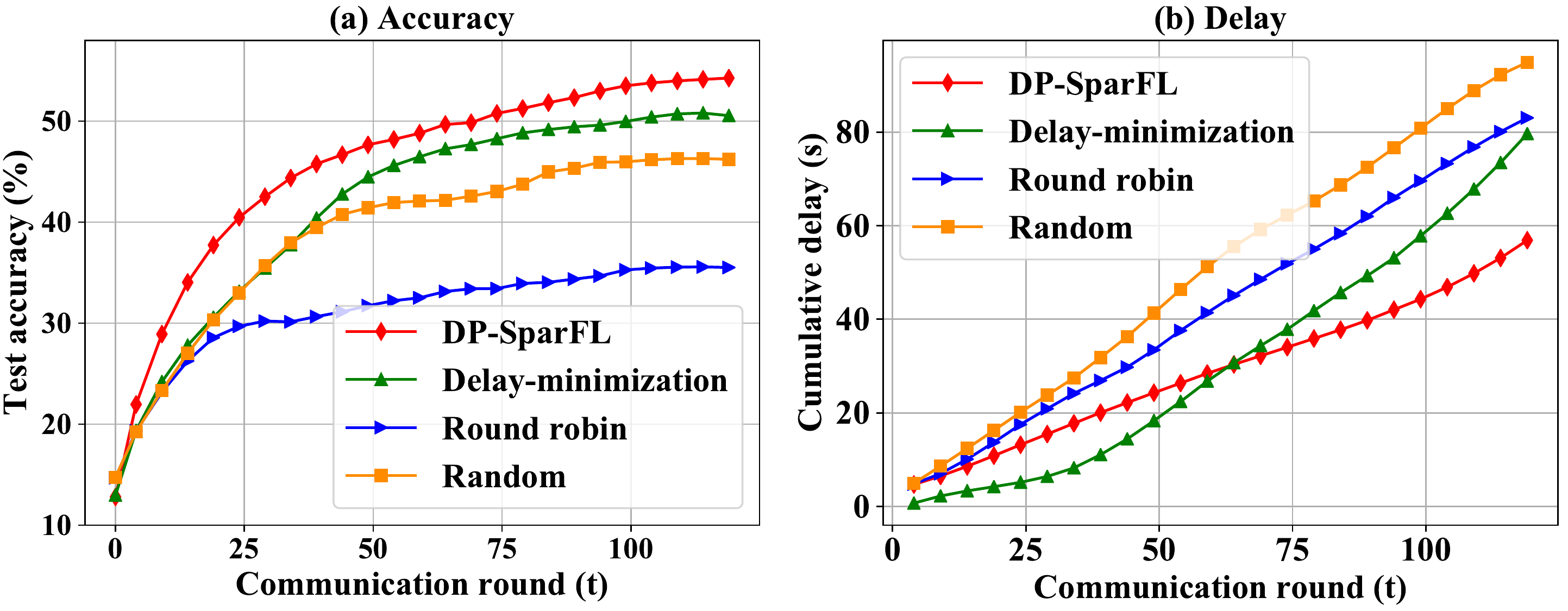}
\caption{The test accuracy and cumulative delay for various algorithms in the non-IID data setting on the FashionMNIST dataset, i.e., DP-SparFL, delay-minimization, round robin and random algorithms.}
\label{fig:acc_delay_noniid_various_algorithms}
\end{figure}
\begin{figure}[ht]
\centering
\includegraphics[width=4.5in,angle=0]{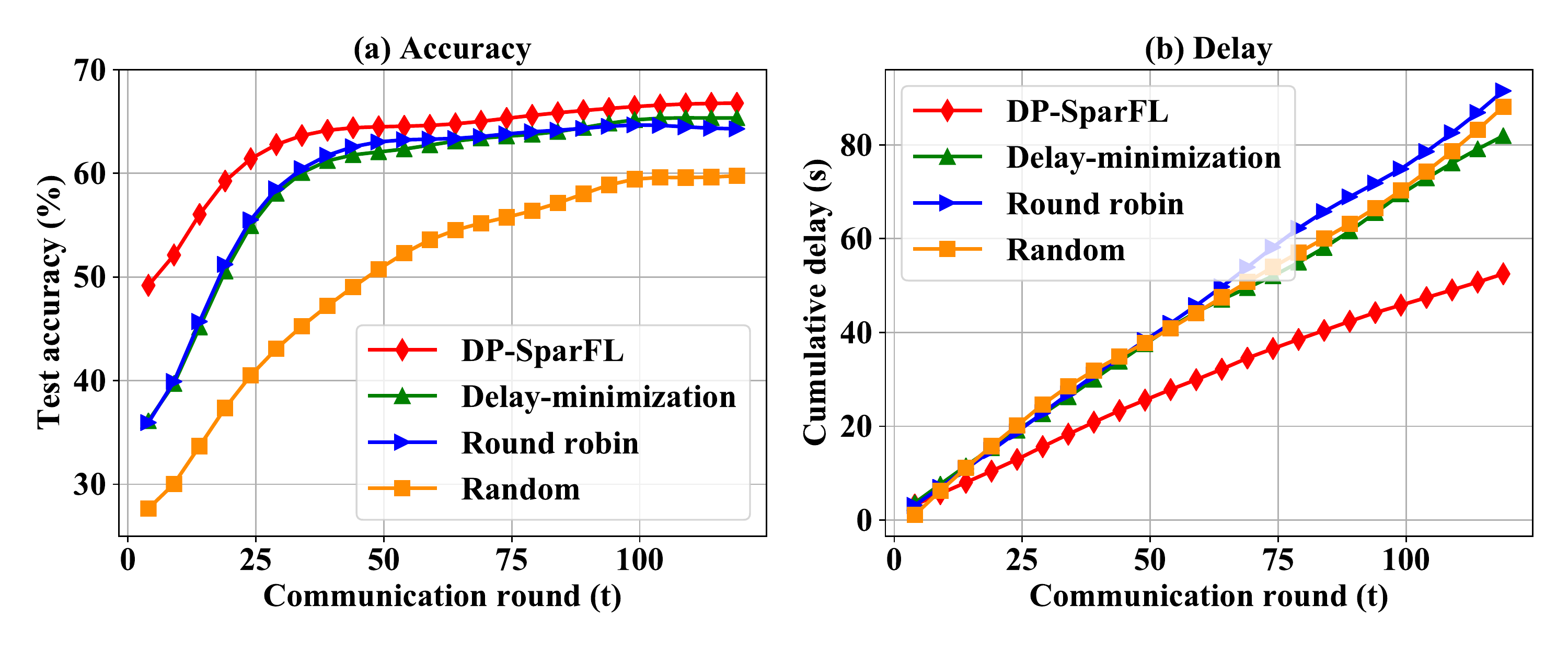}
\caption{The test accuracy and cumulative delay for various algorithms under the imbalance data condition on the FashionMNIST dataset, i.e., DP-SparFL, delay-minimization, round robin and random algorithms.}
\label{fig:acc_delay_imbalance_various_algorithms}
\end{figure}

\section{Conclusions}\label{sec:Concl}
In this paper, we have developed a novel framework for FL over a wireless network with imbalanced resources and DP requirements among clients.
We have proposed a gradient sparsification empowered differentially private FL by an adaptive gradient clipping technique that reduces the connections in gradients of local training for each client before DP noise perturbation.
Then, we have analyzed the convergence bound in terms of gradient sparsification rates by considering a non-convex FL problem.
To further improve the training performance and efficiency, a novel stochastic optimization problem has been formulated to minimize the convergence bound, while satisfying transmit power, average delay and client DP requirement constraints.
We have applied the Lyapunov drift-plus-penalty framework to address the optimization problem.
Our experimental results have exhibited the superiority of our proposed algorithms compared with baselines, i.e., random scheduling, round robin and delay-minimization algorithms.
\appendix
\subsection{Adjusting the Clipping Value}\label{appendix:clipping_value}
Considering the sparsification process, we have $\boldsymbol{g}_{i}^{t}\left(\mathcal{D}_{i,m}\right) = \boldsymbol{g}_{i}^{t}\left(\mathcal{D}_{i,m}\right)\odot\boldsymbol{m}_{i}^{t}$,
where $\mathcal{D}_{i,m}$ is the $m$-th sample of the $i$-th client, $\odot$ represents the element-wise product process, $\boldsymbol{m}_{i}^{t}$ is a binary mask vector $\forall \boldsymbol{m}_{i,k}^{t} \in \{0, 1\}$, $k\in\{1,\ldots,K\}$ and $K = \vert\boldsymbol{m}_{i}^{t} \vert $.
Because probabilities of $\boldsymbol{m}_{i,k}^{t} = 0$ and $\boldsymbol{m}_{i,k}^{t} = 1$ are equal to $s_i^t$ and $(1-s_i^t)$, respectively, we have
\begin{equation}\label{equ:a1}
\begin{aligned}
&(\mathbb{E}[\Vert\boldsymbol{g}_{i}^{t}(\mathcal{D}_{i,m})\odot\boldsymbol{m}_{i}^{t}\Vert])^2\leq\mathbb{E}[\Vert\boldsymbol{g}_{i}^{t}(\mathcal{D}_{i,m})\odot\boldsymbol{m}_{i}^{t}\Vert^2]
=\mathbb{E}\bigg{[}\sum_{k=1}^{K}(\boldsymbol{g}_{i,k}^{t}(\mathcal{D}_{i,m})\odot\boldsymbol{m}_{i,k}^{t})^2\bigg{]}\\
&\quad=\sum_{k'=0}^{K}\frac{k'\binom{K}{k'}}{K}(s_i^{t})^{k'}(1-s_i^{t})^{K-k'}\sum_{k=1}^{K}(\boldsymbol{g}_{i,k}^{t}(\mathcal{D}_{i,m}))^2\\
&\quad=\sum_{k'=0}^{K}\binom{K-1}{k'-1}(s_i^{t})^{k'}(1-s_i^{t})^{K-k'}\sum_{k=1}^{K}(\boldsymbol{g}_{i,k}^{t}(\mathcal{D}_{i,m}))^2=s_i^{t}\Vert\boldsymbol{g}_{i}^{t}(\mathcal{D}_{i,m})\Vert^2,
\end{aligned}
\end{equation}
Therefore, we can obtain
$\mathbb{E}\left[\left\Vert\boldsymbol{g}_{i}^{t}\left(\mathcal{D}_{i,m}\right)\odot\boldsymbol{m}_{i}^{t}\right\Vert\right]\leq \sqrt{s_i^{t}} \left\Vert\boldsymbol{g}_{i}^{t}\left(\mathcal{D}_{i,m}\right)\right\Vert$.
This completes the proof. $\hfill\square$
\subsection{Convergence Analysis}\label{appenidx:conver_analysis}
First, we define $\mathbbm{1}_{i}^{t} \triangleq \sum_{j=1}^{N}\boldsymbol{a}_{i,j}^{t}$ to denote whether the $i$-th client has been allocated to an available channel.
We can note that
\begin{equation}
\begin{aligned}
&\boldsymbol{w}^{t+1}-\boldsymbol{w}^{t} = \sum_{i\in \mathcal{U}}p_{i}^{t}\mathbbm{1}_{i}^{t}\boldsymbol{w}^{t,\tau}_{i}-\boldsymbol{w}^{t,0}_{i}= -\eta\sum_{i\in \mathcal{U}}p_{i}^{t}\mathbbm{1}_{i}^{t}\sum_{\ell =0}^{\tau-1}(\nabla F_{i}(\boldsymbol{w}^{t,\ell}_{i})\odot\boldsymbol{m}_{i}^{t} +\boldsymbol{n}^{t,\ell}_{i}\odot\boldsymbol{m}_{i}^{t}).
\end{aligned}
\end{equation}
Using the $L$-Lipschitz smoothness, we can obtain
\begin{equation}
\begin{aligned}
\mathbbm{E}[F(\boldsymbol{w}^{t+1})- F(\boldsymbol{w}^{t})]&\leq \mathbbm{E}[\langle\nabla F(\boldsymbol{w}^{t}), \boldsymbol{w}^{t+1}-\boldsymbol{w}^{t}\rangle]+ \frac{L}{2}\mathbbm{E}[\Vert \boldsymbol{w}^{t+1}-\boldsymbol{w}^{t} \Vert^{2}]\\
&=\underbrace{-\mathbbm{E}[\langle\nabla F(\boldsymbol{w}^{t}), \eta \sum_{i\in \mathcal{U}}p_{i}^{t}\mathbbm{1}_{i}^{t}\sum_{\ell =0}^{\tau-1}(\nabla F_{i}(\boldsymbol{w}^{t,\ell}_{i})\odot\boldsymbol{m}_{i}^{t} +\boldsymbol{n}^{t,\ell}_{i}\odot\boldsymbol{m}_{i}^{t})\rangle]}_{E_{1}}\\
&\quad+ \frac{\eta^{2}L}{2}\underbrace{\mathbbm{E}[\Vert\sum_{i\in \mathcal{U}}p_{i}^{t}\mathbbm{1}_{i}^{t}\sum_{\ell =0}^{\tau-1}(\nabla F_{i}(\boldsymbol{w}^{t,\ell}_{i})\odot\boldsymbol{m}_{i}^{t} +\boldsymbol{n}^{t,\ell}_{i}\odot\boldsymbol{m}_{i}^{t}) \Vert^{2}]}_{E_2}.
\end{aligned}
\end{equation}
Then, we can rewrite $E_1$ as
\begin{equation}
\begin{aligned}
E_1 &= -\eta\sum_{i\in \mathcal{U}}p_{i}^{t}\mathbbm{1}_{i}^{t}\sum_{\ell =0}^{\tau-1}\langle\nabla F(\boldsymbol{w}^{t}),\mathbbm{E}[\nabla F_{i}(\boldsymbol{w}^{t,\ell}_{i})\odot\boldsymbol{m}_{i}^{t}]\rangle-\eta\sum_{i\in \mathcal{U}}p_{i}^{t}\mathbbm{1}_{i}^{t}\sum_{\ell =0}^{\tau-1}\langle\nabla F(\boldsymbol{w}^{t}),\mathbbm{E}[\boldsymbol{n}^{t,\ell}_{i}\odot\boldsymbol{m}_{i}^{t}]\rangle.
\end{aligned}
\end{equation}
Because $\mathbbm{E}[\boldsymbol{n}^{t,\ell}_{i}]=0$ and $\langle \boldsymbol{x}, \boldsymbol{y}\rangle = \frac{1}{2}(\Vert \boldsymbol{y}\Vert^2+\Vert \boldsymbol{x}\Vert^2-\Vert \boldsymbol{x}-\boldsymbol{y}\Vert^2)$, we have
\begin{equation}
\begin{aligned}
E_1 & = -\frac{\eta}{2}\sum_{\ell =0}^{\tau-1}\Vert \nabla F(\boldsymbol{w}^{t})\Vert^2 - \frac{\eta}{2}\sum_{i\in \mathcal{U}}p_{i}^{t}\mathbbm{1}_{i}^{t}\sum_{\ell =0}^{\tau-1}\mathbbm{E}[\Vert \nabla F_{i}(\boldsymbol{w}^{t,\ell}_{i})\odot\boldsymbol{m}_{i}^{t}\Vert^2]\\
&\quad+\frac{\eta}{2}\sum_{i\in \mathcal{U}}p_{i}^{t}\mathbbm{1}_{i}^{t}\sum_{\ell =0}^{\tau-1}\mathbbm{E}[\Vert \nabla F(\boldsymbol{w}^{t})-\nabla F_{i}(\boldsymbol{w}^{t,\ell}_{i})\odot\boldsymbol{m}_{i}^{t}\Vert^2].
\end{aligned}
\end{equation}
Further, we have
\begin{equation}
\begin{aligned}
E_1 &= -\frac{\eta}{2}\sum_{\ell =0}^{\tau-1}\Vert \nabla F(\boldsymbol{w}^{t})\Vert^2- \frac{\eta}{2}\sum_{i\in \mathcal{U}}p_{i}^{t}\mathbbm{1}_{i}^{t}\sum_{\ell =0}^{\tau-1}\mathbbm{E}[\Vert \nabla F_{i}(\boldsymbol{w}^{t,\ell}_{i})\odot\boldsymbol{m}_{i}^{t}\Vert^2]\\
&\quad+\frac{\eta}{2}\sum_{i\in \mathcal{U}}p_{i}^{t}\mathbbm{1}_{i}^{t}\sum_{\ell =0}^{\tau-1}\mathbbm{E}[\Vert \nabla F(\boldsymbol{w}^{t})-\nabla F_i(\boldsymbol{w}^{t})+\nabla F_i(\boldsymbol{w}^{t})-\nabla F_{i}(\boldsymbol{w}^{t,\ell}_{i})\\
&\quad+\nabla F_{i}(\boldsymbol{w}^{t,\ell}_{i})-\nabla F_{i}(\boldsymbol{w}^{t,\ell}_{i})\odot\boldsymbol{m}_{i}^{t}\Vert^2].
\end{aligned}
\end{equation}
Due to Jensen's inequality and~\eqref{equ:a1} we obtain
\begin{equation}
\begin{aligned}
E_1&\leq-\frac{\eta}{2}\sum_{\ell =0}^{\tau-1}\Vert \nabla F(\boldsymbol{w}^{t})\Vert^2- \frac{\eta}{2}\sum_{i\in \mathcal{U}}p_{i}^{t}\mathbbm{1}_{i}^{t}s_i^{t}\sum_{\ell =0}^{\tau-1}\mathbbm{E}[\Vert \nabla F_{i}(\boldsymbol{w}^{t,\ell}_{i})\Vert^2]
+\frac{3\eta\tau}{2}\sum_{i\in \mathcal{U}}p_{i}^{t}\mathbbm{1}_{i}^{t}\varepsilon_{i}\\
&\quad+\frac{3\eta L^2}{2}\sum_{i\in \mathcal{U}}p_i^{t}\mathbbm{1}_{i}^{t}\sum_{\ell =0}^{\tau-1}\underbrace{\mathbbm{E}[\Vert \boldsymbol{w}^{t}-\boldsymbol{w}^{t,\ell}_{i}\Vert^2]}_{E_{11}}+\frac{3\eta}{2}\sum_{i\in \mathcal{U}}p_{i}^{t}\mathbbm{1}_{i}^{t}(1-s_{i}^{t})\sum_{\ell =0}^{\tau-1}\mathbbm{E}[\Vert \nabla F_{i}(\boldsymbol{w}^{t,\ell}_{i})\Vert^2].
\end{aligned}
\end{equation}
Then, we can bound $E_{11}$ as
\begin{equation}
\begin{aligned}
E_{11}
&= \mathbbm{E}[\Vert \boldsymbol{w}^{t}-(\boldsymbol{w}^{t,\ell-1}_{i}-\eta\nabla F_i(\boldsymbol{w}^{t,\ell-1}_{i})\odot\boldsymbol{m}_{i}^{t}-\eta\boldsymbol{n}^{t,\ell-1}_{i}\odot\boldsymbol{m}_{i}^{t})\Vert^2]\\
&=\eta^2\mathbbm{E}[\Vert \sum_{\kappa = 0}^{\ell-1} (\nabla F_i(\boldsymbol{w}^{t,\kappa}_{i})\odot\boldsymbol{m}_{i}^{t}+\boldsymbol{n}^{t,\kappa}_{i}\odot\boldsymbol{m}_{i}^{t})\Vert^2]\leq \eta^2\ell\sum_{\kappa = 0}^{\ell-1}\mathbbm{E}[\Vert  \nabla F_i(\boldsymbol{w}^{t,\kappa}_{i})\odot\boldsymbol{m}_{i}^{t}+\boldsymbol{n}^{t,\kappa}_{i}\odot\boldsymbol{m}_{i}^{t}\Vert^2]\\
&=\eta^2\ell s_i^{t}\sum_{\kappa = 0}^{\ell-1}\mathbbm{E}[\Vert  \nabla F_i(\boldsymbol{w}^{t,\kappa}_{i})\Vert^2]+\eta^2\ell  s_i^{t}\sum_{\kappa = 0}^{\ell-1}\mathbbm{E}[\Vert\boldsymbol{n}^{t,\kappa}_{i}\Vert^2].
\end{aligned}
\end{equation}
Hence, we can obtain
\begin{equation}
\begin{aligned}
E_1 &\leq-\frac{\eta}{2}\sum_{\ell =0}^{\tau-1}\Vert \nabla F(\boldsymbol{w}^{t})\Vert^2 +\frac{3\eta\tau}{2}\sum_{i\in \mathcal{U}}p_{i}^{t}\mathbbm{1}_{i}^{t}\varepsilon_{i}+\frac{\eta}{2}\sum_{i\in \mathcal{U}}p_{i}^{t}\mathbbm{1}_{i}^{t}(3-4s_i^{t})\sum_{\ell =0}^{\tau-1}\mathbbm{E}[\Vert \nabla F_{i}(\boldsymbol{w}^{t,\ell}_{i})\Vert^2]\\
&\quad+\frac{3\eta^3 L^2}{2}\sum_{i\in \mathcal{U}}p_{i}^{t}\mathbbm{1}_{i}^{t}s_i^{t}\sum_{\ell =0}^{\tau-1}\ell\sum_{\kappa = 0}^{\ell-1}\mathbbm{E}[\Vert  \nabla F_i(\boldsymbol{w}^{t,\kappa}_{i})\Vert^2]+\frac{3\eta^3 L^2}{2}\sum_{i\in \mathcal{U}}p_{i}^{t}\mathbbm{1}_{i}^{t}s_{i}^{t}\sum_{\ell =0}^{\tau-1}\ell\sum_{\kappa = 0}^{\ell-1}\mathbbm{E}[\Vert\boldsymbol{n}^{t,\kappa}_{i}\Vert^2]
\end{aligned}
\end{equation}
and
\begin{equation}
\begin{aligned}
E_{2} &\leq \tau\sum_{i\in \mathcal{U}}p_{i}^{t}\mathbbm{1}_{i}^{t}\sum_{\ell =0}^{\tau-1}\mathbbm{E}[\Vert\nabla F_{i}(\boldsymbol{w}^{t,\ell}_{i})\odot\boldsymbol{m}_{i}^{t}+\boldsymbol{n}^{t,\ell}_{i}\odot\boldsymbol{m}_{i}^{t} \Vert^{2}] \\
&= \tau\sum_{i\in \mathcal{U}}p_{i}^{t}\mathbbm{1}_{i}^{t}s_{i}^{t}\sum_{\ell =0}^{\tau-1}\mathbbm{E}[\Vert\nabla F_{i}(\boldsymbol{w}^{t,\ell}_{i})\Vert^{2}]+\tau\sum_{i\in \mathcal{U}}p_{i}^{t}\mathbbm{1}_{i}^{t}s_{i}^{t}\sum_{\ell =0}^{\tau-1}\mathbbm{E}[\Vert\boldsymbol{n}^{t,\ell}_{i}\Vert^{2}].
\end{aligned}
\end{equation}
Combining $E_1$ and $E_2$, we can obtain
\begin{equation}
\begin{aligned}
&\mathbbm{E}[F(\boldsymbol{w}^{t+1})- F(\boldsymbol{w}^{t})]
\leq
-\frac{\eta}{2}\sum_{\ell =0}^{\tau-1}\Vert \nabla F(\boldsymbol{w}^{t})\Vert^2 +\frac{\eta}{2}\sum_{i\in \mathcal{U}}p_{i}^{t}\mathbbm{1}_{i}^{t}(3-4s_i^{t})\sum_{\ell =0}^{\tau-1}\mathbbm{E}[\Vert \nabla F_{i}(\boldsymbol{w}^{t,\ell}_{i})\Vert^2]\\
&\,+\frac{3\eta^3 L^2}{2}\sum_{i\in \mathcal{U}}p_{i}^{t}\mathbbm{1}_{i}^{t}s_i^{t}\sum_{\ell =0}^{\tau-1}\ell\sum_{\kappa = 0}^{\ell-1}\mathbbm{E}[\Vert  \nabla F_i(\boldsymbol{w}^{t,\kappa}_{i})\Vert^2]+\frac{3\eta^3 L^2}{2}\sum_{i\in \mathcal{U}}p_{i}^{t}\mathbbm{1}_{i}^{t}s_{i}^{t}\sum_{\ell =0}^{\tau-1}\ell\sum_{\kappa = 0}^{\ell-1}\mathbbm{E}[\Vert\boldsymbol{n}^{t,\kappa}_{i}\Vert^2]\\
&\,+\frac{\eta^{2}L\tau}{2}\sum_{i\in \mathcal{U}}p_{i}^{t}\mathbbm{1}_{i}^{t}s_{i}^{t}\sum_{\ell =0}^{\tau-1}\mathbbm{E}[\Vert\nabla F_{i}(\boldsymbol{w}^{t,\ell}_{i}) \Vert^{2}]+\frac{\eta^{2}L\tau}{2}\sum_{i\in \mathcal{U}}p_{i}^{t}\mathbbm{1}_{i}^{t}s_{i}^{t}\sum_{\ell =0}^{\tau-1}\mathbbm{E}[\Vert\boldsymbol{n}^{t,\ell}_{i} \Vert^{2}]+\frac{3\eta\tau}{2}\sum_{i\in \mathcal{U}}p_{i}^{t}\mathbbm{1}_{i}^{t}\varepsilon_{i}.
\end{aligned}
\end{equation}
To ensure the training performance, we will select a proper DP noise variance to have $ \mathbbm{E}[\Vert\boldsymbol{n}^{t,\ell}_{i} \Vert^{2}] = \mathbbm{E}[\Vert\boldsymbol{n}^{t,\kappa}_{i} \Vert^{2}]\leq \Theta$.
Due to the bounded gradient, by setting $\eta L\tau < 1$ and $\eta^3 L^2 \ll 1$, we obtain
\begin{equation}
\begin{aligned}
\mathbbm{E}[F(\boldsymbol{w}^{t+1})- F(\boldsymbol{w}^{t})]
&\leq-\frac{\eta\tau}{2}\Vert \nabla F(\boldsymbol{w}^{t})\Vert^2+ \frac{3\eta\tau G^2}{2}\sum_{i\in \mathcal{U}}p_{i}^{t}\mathbbm{1}_{i}^{t}(1-s_i^{t})\\
&\quad +\frac{\eta^{2}L\tau^2\Theta(1+3\eta L\tau)}{2}+\frac{3\eta\tau}{2}\mathbbm{E}\bigg{[}\sum_{i\in \mathcal{U}}p_{i}^{t}\mathbbm{1}_{i}^{t}\varepsilon_{i}\bigg{]}.
\end{aligned}
\end{equation}
Rearranging and summing $t$ from $0$ to $T-1$, we have
\begin{equation}
\begin{aligned}
\frac{1}{T}\sum_{t=0}^{T-1}\mathbbm{E}[\Vert \nabla F(\boldsymbol{w}^{t})\Vert^2] &\leq \frac{2(F(\boldsymbol{w}^{0})- F(\boldsymbol{w}^{T}))}{\eta\tau T} + 3\varepsilon\\
&\quad+ \frac{3G^2}{T}\sum_{t=0}^{T-1}\sum_{i\in \mathcal{U}}p_{i}^{t}\mathbbm{1}_{i}^{t}(1-s_i^{t})+\eta \tau^2\Theta(1+3\eta L\tau).
\end{aligned}
\end{equation}
This completes the proof. $\hfill\square$

\subsection{Solution of the Optimal Gradient Sparsification Rate}\label{appendix:optimal_pruning}
To obtain the optimal gradient sparsification rate, we first derive the relation between $s_{i}^{t}$ and $V^t$.
Hence, we first consider the condition $Q^{t, \text{de}}> 0$ and have
\begin{equation}
\begin{aligned}
V^t &= \sum_{i\in \mathcal{U}}\sum_{j\in \mathcal{N}}(Q_{i}^{t, \text{fa}}-\lambda p_{i}^{t} s_{i}^{t})\boldsymbol{a}_{i,j}^{t} + Q^{t, \text{de}}\left(d^{t}-d^{\text{Avg}}\right)-\sum_{i\in \mathcal{U}}\sum_{j\in \mathcal{N}}\beta_{i}
= \sum_{i\in \mathcal{U}}\sum_{j\in \mathcal{N}}(Q_{i}^{t, \text{fa}}-\lambda p_{i}^{t} s_{i}^{t})\boldsymbol{a}_{i,j}^{t}\\
&+Q^{t, \text{de}}\max_{i\in\mathcal{N}}\Bigg{\{}\sum_{j\in \mathcal{N}}\boldsymbol{a}_{i,j}^{t}\Bigg{(}\frac{ Zp_{i}^{t}s_{i}^{t}}{B\log_{2}\left(1+\frac{P_{i}^{t} h_{i,j}^{t}}{\sigma^{2}}\right)}+d_i^{t, \text{do}}+\frac{\tau \vert \mathcal D_{i}\vert \Phi_{i}}{f_{i}^{t}}\Bigg{)}\Bigg{\}}-Q^{t, \text{de}} d^{\text{Avg}}-\sum_{i\in \mathcal{U}}\sum_{j\in \mathcal{N}}\beta_{i}.
\end{aligned}
\end{equation}
Due to the maximizing process, this problem can be divided into $N$ subproblems based on the client with the maximum delay.
First, let us discuss the condition that the delay of the client owning the $j$-th channel is the maximum one among all clients.
We assume the $j$-th channel is allocated to the $i$-th client and its delay is the maximum one.
Thus, we can obtain
\begin{equation}
\begin{aligned}
d_{i, j}^{t}&=\frac{ Zp_{i}^{t}s_{i}^{t}}{B\log_{2}\left(1+\frac{P_{i}^{t} h_{i,j}^{t}}{\sigma^{2}}\right)}+d_i^{t, \text{do}}+\frac{\tau \vert \mathcal D_{i}\vert \Phi_{i}}{f_{i}^{t}}\\
&\geq\max_{i'\in\mathcal{U}/i}\Bigg{\{}\sum_{j'\in \mathcal{N}}\boldsymbol{a}_{i',j'}^{t}\Bigg{(}\frac{ Zp_{i'}^{t}s_{i'}^{t}}{B\log_{2}\Big{(}1+\frac{P_{i'}^{t} h_{i',j'}^{t}}{\sigma^{2}}\Big{)}}+d_{i'}^{t, \text{do}}+\frac{\tau \vert \mathcal D_{i'}\vert \Phi_{i'}}{f_{i'}^{t}}\Bigg{)}\Bigg{\}}.
\end{aligned}
\end{equation}
From above inequation, we have
\begin{equation}
\begin{aligned}
s_{i}^{t}\geq s_{i}^{t, \text{min}}&\triangleq B\log_{2}\left(1+\frac{P_{i}^{t} h_{i,j}^{t}}{\sigma^{2}}\right)\\
&\cdot\frac{\max\limits_{i'\in\mathcal{U}/i}\Bigg{\{}\sum_{j'\in \mathcal{N}}\boldsymbol{a}_{i',j'}^{t}\Bigg{(}\frac{ Zp_{i'}^{t}s_{i'}^{\text{th}}}{B\log_{2}\Big{(}1+\frac{P_{i'}^{t} h_{i',j'}^{t}}{\sigma^{2}}\Big{)}}+d_{i'}^{t, \text{do}}+\frac{\tau \vert \mathcal D_{i'}\vert \Phi_{i'}}{f_{i'}^{t}}\Bigg{)}\Bigg{\}}-d_i^{t, \text{do}}-\frac{\tau \vert \mathcal D_{i}\vert \Phi_{i}}{f_{i}^{t}}}{Z p_{i}^{t}}.
\end{aligned}
\end{equation}
If $s_{i}^{t, \text{min}} > 1$, this subproblem does not have solutions.
Otherwise, we can derive the first derivative of $V^t$ with respect to $s_{i}^{t}$ of the $j$-th subproblem as follows:
\begin{equation}
\begin{aligned}
\frac{\partial V^t}{\partial s_{i}^{t}} =
-\lambda + \frac{Zp_{i}^{t}Q^{t, \text{de}}}{B\log_{2}\left(1+\frac{P_{i}^{t}  h_{i,j}^{t}}{\sigma^{2}}\right)}.
\end{aligned}
\end{equation}
If $\frac{Zp_{i}^{t}Q^{t, \text{de}}}{B\log_{2}\left(1+\frac{P_{i}^{t} h_{i,j}^{t}}{\sigma^{2}}\right)} \leq \lambda$, it can be found that as the value of $s_{i}^{t}$ increases, the objective $V^t$ decreases.
Hence, we have $s_i^{t,*} = 1 $.
For other fast clients, i.e., $i'\in\mathcal{U}/i$, we have
\begin{equation}
\begin{aligned}
&s_{i'}^{t}\leq s_{i'}^{\text{max}}\triangleq\frac{Zp_{i'}^{t}\Bigg{(}\frac{ Zp_{i}^{t}s_{i}^{t}}{B\log_{2}\left(1+\frac{P_{i}^{t} h_{i,j}^{t}}{\sigma^{2}}\right)}+d_{i}^{t, \text{do}}+\frac{\tau \vert \mathcal D_{i}\vert \Phi_{i}}{f_{i}^{t}}-d_{i'}^{t, \text{do}}-\frac{\tau \vert \mathcal D_{i'}\vert \Phi_{i}}{f_{i'}^{t}}\Bigg{)}}{\sum_{j'=1}^{N}\boldsymbol{a}_{i',j'}^{t}B\log_{2}\Big{(}1+\frac{P_{i'}^{t} h_{i',j'}^{t}}{\sigma^{2}}\Big{)}}.
\end{aligned}
\end{equation}
We can also derive the first derivative of $V^t$ with respect to $s_{i'}^{t}$ as $\frac{\partial V^t}{\partial s_{i'}^{t}} = -p_{i'}^{t}\lambda$.
Therefore, we have $s_{i'}^{t,*} = \min \{s_{i'}^{\text{max}}, 1\}$.

If $\frac{Zp_{i}^{t}Q^{t, \text{de}}}{B\log_{2}\left(1+\frac{P_{i}^{t} h_{i,j}^{t}}{\sigma^{2}}\right)} > \lambda$, we can find that as the value of $s_{i}^{t}$ increases, the objective $V^t$ increases.
Therefore, the system need to select a small gradient sparsification rate $s_{i}^{t}$ in $[s^{\text{th}} , 1]$ for the $i$-th client.
However, for the $i'$-th client, i.e., $i'\in\mathcal{U}/i$, we want to select a large gradient sparsification rate in $[s^{\text{th}} , \min\{1, s_{i'}^{\text{max}}\}]$  because the first derivative of $V^t$ with respect to $s_{i'}^{t}$ is negative.
We can decrease the $s_{i}^{t}$ from $1$ and then $s_{i'}^{\text{max}}$ may be selected.
Therefore, the first derivative of $V^t$ with respect to $s_{i}^{t}$ should be modified because $s_{i'}^{\text{max}}$ is related to $s_{i}^{t}$.
When the first derivative of $V^t$ with respect to $s_{i}^{t}$ become negative, let us stop decreasing the value of $s_{i}^{t}$.
We can note that this way can obtain the optimal $s_{i}^{t,*}$ and $s_{i'}^{t,*} = \min \{s_{i'}^{\text{max}}, 1\}$.

Other subproblems can be addressed using the same method.
Overall, after addressing all $N$ subproblems, the optimal solution can be obtained as the final output.\\
This completes the proof. $\hfill\square$
\subsection{Feasibility Analysis}\label{appendix:feasibility_analysis}
We first introduce the Lyapunov function $\Gamma(\boldsymbol{Q}^t)=\frac{1}{2}(Q^{t, \text{de}})^{2}+\frac{1}{2}\sum_{i\in\mathcal{U}}(Q_{i}^{t, \text{fa}})^{2}$, in which the drift from one communication round can be given as
\begin{equation}
\begin{aligned}
&\Gamma\left(\boldsymbol{Q}^{t+1}\right)-\Gamma\left(\boldsymbol{Q}^{t}\right)=\frac{1}{2}\left(\max\left\{Q^{t, \text{de}}+d^{\text{Avg}}-d^{t}, 0\right\}\right)^{2}-\frac{1}{2}\left(Q^{t, \text{de}}\right)^{2}\\
&\quad\quad+\frac{1}{2}\sum_{i\in\mathcal{U}}\left(\max\left\{Q^{t-1, \text{fa}}_{i}+\mathbbm{1}_{i}^{t}-\beta_{i}, 0\right\}\right)^{2}-\frac{1}{2}\sum_{i\in\mathcal{U}}(Q_{i}^{t, \text{fa}})^{2}\\
&\quad\leq\frac{1}{2}\left(d^{\text{Avg}}-d^{t}\right)^{2}+Q^{t, \text{de}}\left(d^{\text{Avg}}-d^{t}\right)+\frac{1}{2}\sum_{i\in\mathcal{U}}\left(\mathbbm{1}_{i}^{t}-\beta_{i}\right)^{2}
+\sum_{i\in\mathcal{U}}Q_{i}^{t, \text{fa}}\left(\mathbbm{1}_{i}^{t}-\beta_{i}\right).
\end{aligned}
\end{equation}
Because
\begin{equation}
\begin{aligned}
&V^t(\boldsymbol{P}^{t}, \boldsymbol{s}^{t}, \boldsymbol{a}^{t})= -\lambda \sum_{i\in \mathcal{U}}p_{i}^{t}\mathbbm{1}_{i}^{t}s_{i}^{t}+\sum_{i\in \mathcal{U}}Q_{i}^{t, \text{fa}}(\mathbbm{1}_{i}^{t}-\beta_{i}) + Q^{t, \text{de}}\left(d^{t}-d^{\text{Avg}}\right),
\end{aligned}
\end{equation}
we have
\begin{equation}
\begin{aligned}
V^t(\boldsymbol{P}^{t}, \boldsymbol{s}^{t}, \boldsymbol{a}^{t})
&\leq
-\lambda \sum_{i\in \mathcal{U}}p_{i}^{t}\mathbbm{1}_{i}^{t}s_{i}^{t}+\frac{1}{2}\left(d^{\text{Avg}}-d^{t}\right)^{2}+Q^{t, \text{de}}\left(d^{\text{Avg}}-d^{t}\right)\\
&\quad+\frac{1}{2}\sum_{i\in\mathcal{U}}\left(\mathbbm{1}_{i}^{t}-\beta_{i}\right)^{2}+\sum_{i\in\mathcal{U}}Q_{i}^{t, \text{fa}}\left(\mathbbm{1}_{i}^{t}-\beta_{i}\right).
\end{aligned}
\end{equation}
Due to $s_{i}^{t}, p_{i}^{t}\leq 1$, $\boldsymbol{a}_{i,j}^{t}\in \{0,1\}$, $\beta_{i}\leq \frac{N}{U}$, $d^{t}=\max_{i\in\mathcal{U}} \mathbbm{1}_{i}^{t} d_{i, j}^{t}$ and $d_{i, j}^{t} = d_{i}^{t, \text{do}} + d_{i}^{t, \text{lo}} + d_{i, j}^{t, \text{up}}$, $\forall i \in \mathcal{U}, j\in \mathcal{N}$,
we have
\begin{equation}
\begin{aligned}
&\mathbb{E}\left[V^t(\boldsymbol{P}^{t}, \boldsymbol{s}^{t}, \boldsymbol{a}^{t})\vert \boldsymbol{Q}^{t}\right]\leq C_1 + Q^{t, \text{de}}\mathbb{E}\left[d^{\text{Avg}}-d^{t}\vert \boldsymbol{Q}^{t}\right] + \sum_{i\in\mathcal{U}}Q_{i}^{t, \text{fa}}\mathbb{E}\left[\mathbbm{1}_{i}^{t}-\beta_{i}\vert \boldsymbol{Q}^{t}\right].
\end{aligned}
\end{equation}
Based on Theorem 4.5 in~\cite{Neely2010Stochastic} and Lemma 1 in~\cite{Deng2020Wireless}, existing $\zeta^{\text{opt}} > 0$, we can obtain the following inequality:
\begin{equation}
\begin{aligned}
&\mathbb{E}\left[V^t(\boldsymbol{P}^{t}, \boldsymbol{s}^{t}, \boldsymbol{a}^{t})\right]\leq C_1 + \zeta^{\text{opt}}.
\end{aligned}
\end{equation}
By summing this equation over $t = 0, 1, \ldots, T$, we obtain
\begin{equation}
\begin{aligned}
\lim_{T\rightarrow\infty}\sup\frac{1}{T}\sum_{t=0}^{T-1}\mathbb{E}\left[V^t(\boldsymbol{P}^{t}, \boldsymbol{s}^{t}, \boldsymbol{a}^{t})\right]\leq C_1 + \zeta^{\text{opt}}<\infty.
\end{aligned}
\end{equation}
This completes the proof. $\hfill\square$

\bibliography{reference}

\begin{thebibliography}{10}
\providecommand{\url}[1]{#1}
\csname url@samestyle\endcsname
\providecommand{\newblock}{\relax}
\providecommand{\bibinfo}[2]{#2}
\providecommand{\BIBentrySTDinterwordspacing}{\spaceskip=0pt\relax}
\providecommand{\BIBentryALTinterwordstretchfactor}{4}
\providecommand{\BIBentryALTinterwordspacing}{\spaceskip=\fontdimen2\font plus
\BIBentryALTinterwordstretchfactor\fontdimen3\font minus
  \fontdimen4\font\relax}
\providecommand{\BIBforeignlanguage}[2]{{%
\expandafter\ifx\csname l@#1\endcsname\relax
\typeout{** WARNING: IEEEtran.bst: No hyphenation pattern has been}%
\typeout{** loaded for the language `#1'. Using the pattern for}%
\typeout{** the default language instead.}%
\else
\language=\csname l@#1\endcsname
\fi
#2}}
\providecommand{\BIBdecl}{\relax}
\BIBdecl

\bibitem{Nguyen2021Enabling}
D.~C. Nguyen \emph{et~al.}, ``Enabling {AI} in future wireless networks: {A}
  data life cycle perspective,'' \emph{{IEEE} Commun. Surv. Tutorials},
  vol.~23, no.~1, pp. 553--595, 2021.

\bibitem{Kairouz2019Advances}
P.~{Kairouz} \emph{et~al.}, ``Advances and open problems in federated
  learning,'' \emph{Found. Trends Mach. Learn.}, vol.~14, no. 1-2, pp. 1--210,
  2021.

\bibitem{Ma2019FL}
C.~{Ma}, J.~{Li}, M.~{Ding}, H.~H. {Yang}, F.~{Shu}, T.~Q.~S. {Quek}, and H.~V.
  {Poor}, ``On safeguarding privacy and security in the framework of federated
  learning,'' \emph{{IEEE} Netw.}, vol.~34, no.~4, pp. 242--248, 2020.

\bibitem{Xia2020Multi}
W.~Xia, T.~Q.~S. Quek, K.~Guo, W.~Wen, H.~H. Yang, and H.~Zhu, ``Multi-armed
  bandit-based client scheduling for federated learning,'' \emph{IEEE Trans.
  Wireless Commun.}, vol.~19, no.~11, pp. 7108--7123, 2020.

\bibitem{Yang2020Scheduling}
H.~H. Yang, Z.~Liu, T.~Q.~S. Quek, and H.~V. Poor, ``Scheduling policies for
  federated learning in wireless networks,'' \emph{IEEE Trans. Commun.},
  vol.~68, no.~1, pp. 317--333, 2020.

\bibitem{Wang2019Adaptive}
S.~{Wang}, T.~{Tuor}, T.~{Salonidis}, K.~K. {Leung}, C.~{Makaya}, T.~{He}, and
  K.~{Chan}, ``Adaptive federated learning in resource constrained edge
  computing systems,'' \emph{IEEE J. Sel. Areas Commun.}, vol.~37, no.~6, pp.
  1205--1221, 2019.

\bibitem{Chen2020A}
M.~{Chen}, Z.~{Yang}, W.~{Saad}, C.~{Yin}, H.~V. {Poor}, and S.~{Cui}, ``A
  joint learning and communications framework for federated learning over
  wireless networks,'' \emph{IEEE Trans. Wireless Commun.}, vol.~20, no.~1, pp.
  269--283, 2021.

\bibitem{Chen2020Convergence}
M.~Chen, H.~V. Poor, W.~Saad, and S.~Cui, ``Convergence time optimization for
  federated learning over wireless networks,'' \emph{IEEE Trans. Wireless
  Commun.}, vol.~20, no.~4, pp. 2457--2471, 2021.

\bibitem{yang2021Energy}
Z.~Yang, M.~Chen, W.~Saad, C.~S. Hong, and M.~Shikh-Bahaei, ``Energy efficient
  federated learning over wireless communication networks,'' \emph{IEEE Trans.
  Wireless Commun.}, vol.~20, no.~3, pp. 1935--1949, 2021.

\bibitem{Deng2022Blockchain}
X.~Deng, J.~Li, C.~Ma, K.~Wei, L.~Shi, M.~Ding, W.~Chen, and H.~Vincent~Poor,
  ``Blockchain assisted federated learning over wireless channels: {Dynamic}
  resource allocation and client scheduling,'' \emph{IEEE Trans. Wireless
  Commun.}, Early Access 2022.

\bibitem{Deng2022Low}
X.~Deng, J.~Li, C.~Ma, K.~Wei, L.~Shi, M.~Ding, and W.~Chen, ``Low-latency
  federated learning with dnn partition in distributed industrial iot
  networks,'' \emph{{IEEE} J. Sel. Areas Commun.}, Early Access 2022.

\bibitem{Decebal2019Scalable}
D.~C. Mocanu, E.~Mocanu, P.~Stone, P.~H. Nguyen, M.~Gibescu, and A.~Liotta,
  ``Scalable training of artificial neural networks with adaptive sparse
  connectivity inspired by network science,'' \emph{Nat. Commun.}, vol.~9,
  2018.

\bibitem{Jiang2018SketchML}
J.~Jiang, F.~Fu, T.~Yang, and B.~Cui, ``{SketchML}: {Accelerating} distributed
  machine learning with data sketches,'' in \emph{Proc. International
  Conference on Management of Data (SIGMOD)}, Houston, TX, USA, June 2018, pp.
  1269--1284.

\bibitem{Chen2021Communication}
M.~Chen, N.~Shlezinger, H.~V. Poor, Y.~C. Eldar, and S.~Cui,
  ``Communication-efficient federated learning,'' \emph{Proc. Natl. Acad. Sci.
  U.S.A.}, vol. 118, no.~17, Apr. 2021.

\bibitem{Zheng2021Design}
S.~Zheng, C.~Shen, and X.~Chen, ``Design and analysis of uplink and downlink
  communications for federated learning,'' \emph{IEEE J. Sel. Areas Commun.},
  vol.~39, no.~7, pp. 2150--2167, 2021.

\bibitem{Wang2022Quantized}
Y.~Wang, Y.~Xu, Q.~Shi, and T.~Chang, ``Quantized federated learning under
  transmission delay and outage constraints,'' \emph{{IEEE} J. Sel. Areas
  Commun.}, vol.~40, no.~1, pp. 323--341, 2022.

\bibitem{Liu2022Joint}
S.~Liu, G.~Yu, R.~Yin, J.~Yuan, L.~Shen, and C.~Liu, ``Joint model pruning and
  device selection for communication-efficient federated edge learning,''
  \emph{{IEEE} Trans. Commun.}, vol.~70, no.~1, pp. 231--244, 2022.

\bibitem{Wang2019Beyond}
Z.~{Wang} \emph{et~al.}, ``Beyond inferring class representatives: User-level
  privacy leakage from federated learning,'' in \emph{Proc. IEEE International
  Conference on Computer Communications (INFOCOM)}, Paris, France, Apr. 2019,
  pp. 2512--2520.

\bibitem{Dwork2014The}
C.~Dwork and A.~Roth, \emph{The Algorithmic Foundations of Differential
  Privacy}, 2014, vol.~9, no. 3-4.

\bibitem{Wei2020Fed}
K.~{Wei} \emph{et~al.}, ``Federated learning with differential privacy:
  {Algorithms} and performance analysis,'' \emph{{IEEE} Trans. Inf. Forensics
  Secur.}, vol.~15, pp. 3454--3469, 2020.

\bibitem{Wei2021User}
K.~{Wei}, J.~{Li}, M.~{Ding}, C.~{Ma}, H.~{Su}, B.~{Zhang}, and H.~V. {Poor},
  ``User-level privacy-preserving federated learning: {Analysis} and
  performance optimization,'' \emph{IEEE Trans. Mob. Comput.}, vol.~21, no.~9,
  pp. 3388--3401, 2022.

\bibitem{Abadi2016Deep}
A.~Martin \emph{et~al.}, ``Deep learning with differential privacy,'' in
  \emph{Proc. ACM Conference on Computer and Communications Security (CCS)},
  Vienna, Austria, Oct. 2016, pp. 308--318.

\bibitem{Zhou2021Bypassing}
Y.~Zhou, S.~Wu, and A.~Banerjee, ``Bypassing the ambient dimension: Private
  {SGD} with gradient subspace identification,'' in \emph{Proc. International
  Conference on Learning Representations (ICLR)}, Virtual, May 2021.

\bibitem{Yang2021Privacy}
H.~Yang, J.~Zhao, Z.~Xiong, K.-Y. Lam, S.~Sun, and L.~Xiao,
  ``Privacy-preserving federated learning for {UAV}-enabled networks:
  {Learning}-based joint scheduling and resource management,'' \emph{{IEEE} J.
  Sel. Areas Commun.}, vol.~39, no.~10, pp. 3144--3159, 2021.

\bibitem{Wei2022Low}
K.~Wei, J.~Li, C.~Ma, M.~Ding, C.~Chen, S.~Jin, Z.~Han, and H.~V. Poor,
  ``Low-latency federated learning over wireless channels with differential
  privacy,'' \emph{IEEE J. Sel. Areas Commun.}, vol.~40, no.~1, pp. 290--307,
  2022.

\bibitem{Mironov2019R}
\BIBentryALTinterwordspacing
I.~Mironov, K.~Talwar, and L.~Zhang, ``R{\'{e}}nyi differential privacy of the
  sampled gaussian mechanism,'' \emph{Arxiv}, 2019. [Online]. Available:
  \url{http://arxiv.org/abs/1908.10530}
\BIBentrySTDinterwordspacing

\bibitem{Yousefpour2021Opacus}
A.~Yousefpour, I.~Shilov, A.~Sablayrolles, D.~Testuggine, K.~Prasad, M.~Malek,
  J.~Nguyen, S.~Gosh, A.~Bharadwaj, J.~Zhao, G.~Cormode, and I.~Mironov,
  ``Opacus: {User}-friendly differential privacy library in pytorch,'' in
  \emph{Proc. Privacy in Machine Learning (PriML) Workshop, NeurIPS}, Virtual,
  Dec. 2021.

\bibitem{Stich2018Sparsified}
S.~U. Stich, J.~Cordonnier, and M.~Jaggi, ``Sparsified {SGD} with memory,'' in
  \emph{Proc. Advances in Neural Information Processing Systems (NeurIPS)},
  Montr{\'{e}}al, Canada, Dec. 2018, pp. 4452--4463.

\bibitem{Dinh2021Federated}
C.~T. Dinh \emph{et~al.}, ``Federated learning over wireless networks:
  Convergence analysis and resource allocation,'' \emph{IEEE/ACM Trans. Netw.},
  vol.~29, no.~1, pp. 398--409, 2021.

\bibitem{Huang2021An}
T.~Huang, W.~Lin, W.~Wu, L.~He, K.~Li, and A.~Y. Zomaya, ``An
  efficiency-boosting client selection scheme for federated learning with
  fairness guarantee,'' \emph{IEEE Trans. Parallel Distrib. Syst.}, vol.~32,
  no.~7, pp. 1552--1564, 2021.

\bibitem{Pavlo2019Importance}
P.~Molchanov, A.~Mallya, S.~Tyree, I.~Frosio, and J.~Kautz, ``Importance
  estimation for neural network pruning,'' in \emph{Proc. {IEEE} Conference on
  Computer Vision and Pattern Recognition (CVPR)}, Long Beach, USA, June 2019,
  pp. 11\,264--11\,272.

\bibitem{Neely2010Stochastic}
M.~J. Neely, \emph{Stochastic Network Optimization with Application to
  Communication and Queueing Systems}.\hskip 1em plus 0.5em minus 0.4em\relax
  Synthesis Lectures on Communication Networks, 2010.

\bibitem{Kang2018Low}
S.~Kang and C.~Joo, ``Low-complexity learning for dynamic spectrum access in
  multi-user multi-channel networks,'' in \emph{Proc. IEEE Conference on
  Computer Communications (INFOCOM)}, Honolulu, HI, USA, Apr. 2018, pp.
  1367--1375.

\bibitem{Mohammad2011Online}
M.~Mahdian and Q.~Yan, ``Online bipartite matching with random arrivals: an
  approach based on strongly factor-revealing {LP}s,'' in \emph{Proc. {ACM}
  Symposium on Theory of Computing (STOC)}, San Jose, CA, USA, Jun. 2011, pp.
  597--606.

\bibitem{lecun1998gradient}
Y.~LeCun, L.~Bottou, Y.~Bengio, P.~Haffner \emph{et~al.}, ``Gradient-based
  learning applied to document recognition,'' \emph{Proc. IEEE}, vol.~86,
  no.~11, pp. 2278--2324, 1998.

\bibitem{Xiao2017Fashion}
\BIBentryALTinterwordspacing
H.~Xiao, K.~Rasul, and R.~Vollgraf, ``Fashion-mnist: a novel image dataset for
  benchmarking machine learning algorithms,'' \emph{arXiv}, 2017. [Online].
  Available: \url{http://arxiv.org/abs/1708.07747}
\BIBentrySTDinterwordspacing

\bibitem{Krizhevsky2009Learning}
\BIBentryALTinterwordspacing
A.~Krizhevsky and G.~Hinton, ``Learning multiple layers of features from tiny
  images,'' \emph{M.S. thesis, Univ. Toronto}, 2009. [Online]. Available:
  \url{http://www.cs.toronto.edu/~kriz/learning-features-2009-TR.pdf}
\BIBentrySTDinterwordspacing

\bibitem{Yurochkin2019Bayesian}
M.~Yurochkin, M.~Agarwal, S.~Ghosh, K.~Greenewald, N.~Hoang, and Y.~Khazaeni,
  ``{B}ayesian nonparametric federated learning of neural networks,'' in
  \emph{Proc. International Conference on Machine Learning}, Jun. 2019, pp.
  7252--7261.

\bibitem{Liu2021Projected}
J.~Liu, J.~Lou, L.~Xiong, J.~Liu, and X.~Meng, ``Projected federated averaging
  with heterogeneous differential privacy,'' \emph{Proc. VLDB Endow.}, vol.~15,
  no.~4, pp. 828--840, Dec. 2021.

\bibitem{Deng2020Wireless}
X.~{Deng}, J.~{Li}, L.~{Shi}, Z.~{Wei}, X.~{Zhou}, and J.~{Yuan}, ``Wireless
  powered mobile edge computing: Dynamic resource allocation and throughput
  maximization,'' \emph{IEEE Trans. Mob. Comput.}, vol.~21, no.~6, pp.
  2271--2288, 2022.

\end{thebibliography}
\bibliographystyle{IEEEtran}

\end{document}